\documentclass[10pt,reqno,a4paper]{amsart}
\usepackage[british]{babel}
\usepackage{amssymb,amstext,amsthm,eucal,bbm,mathrsfs,amscd,bbold}
\usepackage{array,color}
\usepackage[utf8]{inputenc}
\usepackage{enumitem}
\usepackage[small]{eulervm}
\usepackage{tgpagella}
\usepackage[unicode]{hyperref}
\hypersetup{%
  pdftitle   = {On the algebraic structure of Killing superalgebras},
  pdfkeywords = {Spencer cohomology, Lie superalgebra, Poincaré,
    supergravity, supersymmetry, Killing superalgebra},
  pdfauthor  = {José Figueroa-O'Farrill, Andrea Santi},
  pdfcreator = {\LaTeX\ with package \flqq hyperref\frqq}
}
\theoremstyle{plain}
\newtheorem{lemma}{Lemma}
\newtheorem{proposition}[lemma]{Proposition}
\newtheorem{theorem}[lemma]{Theorem}
\newtheorem{corollary}[lemma]{Corollary}
\newtheorem*{mainthm}{Theorem}
\theoremstyle{definition}
\newtheorem{remark}[lemma]{Remark}
\newtheorem{definition}[lemma]{Definition}
\DeclareMathOperator{\im}{im}
\DeclareMathOperator{\tr}{tr}
\DeclareMathOperator{\Ric}{Ric}
\DeclareMathOperator{\Ad}{Ad}
\newcommand{\Hom}{\mathrm{Hom}}

\newcommand{\GL}{\mathrm{GL}}
\newcommand{\SL}{\mathrm{SL}}
\newcommand{\SO}{\mathrm{SO}}
\newcommand{\Spin}{\mathrm{Spin}}
\newcommand{\CSpin}{\mathrm{CSpin}}
\newcommand{\Cl}{C\ell}

\newcommand{\ad}{\mathrm{ad}}
\newcommand{\ev}{\mathrm{ev}}

\newcommand{\stab}{\mathfrak{stab}}

\newcommand{\fso}{\mathfrak{so}}

\newcommand{\fg}{\mathfrak{g}}
\newcommand{\fh}{\mathfrak{h}}

\newcommand{\fk}{\mathfrak{k}}
\newcommand{\fp}{\mathfrak{p}}

\newcommand{\fa}{\mathfrak{a}}
\newcommand{\fD}{\mathfrak{D}}
\newcommand{\1}{\mathbb{1}}

\newcommand{\RR}{\mathbb{R}}

\newcommand{\ZZ}{\mathbb{Z}}

\newcommand{\eE}{\mathscr{E}}
\newcommand{\eL}{\mathscr{L}}
\newcommand{\eX}{\mathscr{X}}
\newcommand{\vol}{\operatorname{vol}}

\newcommand{\sk}{\operatorname{Skew}}

\allowdisplaybreaks
\begin{document}

\title{On the algebraic structure of Killing superalgebras}
\author{José Figueroa-O'Farrill}
\author{Andrea Santi}
\address{Maxwell Institute and School of Mathematics, The University
  of Edinburgh, James Clerk Maxwell Building, Peter Guthrie Tait Road,
  Edinburgh EH9 3FD, United Kingdom}
\thanks{EMPG-16-13}
\begin{abstract}
  We study the algebraic structure of the Killing superalgebra of a
  supersymmetric background of $11$-dimensional supergravity and show
  that it is isomorphic to a filtered deformation of a $\ZZ$-graded
  subalgebra of the Poincaré superalgebra.  We are able to map the
  classification problem for highly supersymmetric backgrounds (i.e.,
  those which preserve more than half the supersymmetry) to the
  classification problem of a certain class of filtered deformations
  of graded subalgebras of the Poincaré superalgebra.  We show that
  one can reconstruct a highly supersymmetric background from its
  Killing superalgebra; in so doing, we relate the bosonic field
  equations of $11$-dimensional supergravity to the Jacobi identity of
  the Killing superalgebra and show in this way that preserving more
  than half the supersymmetry implies the bosonic field equations.
\end{abstract}
\maketitle
\tableofcontents

\section{Introduction}
\label{sec:introduction}

Arguably the most interesting open problem in eleven-dimensional
supergravity is the classification of (supersymmetric, bosonic)
backgrounds. This problem has a long pedigree. It started in the
1980s, where it took the form of the classification problem for
Freund--Rubin backgrounds (and generalisations thereof) in the context
of Kaluza--Klein supergravity. The substantial progress made during
this time is fairly well documented in the review \cite{DNP}. One
problem with Freund--Rubin backgrounds from a Kaluza--Klein
perspective is that the spacetime and the compact extra dimensions
have commensurate radii of curvature, but they resurfaced in the 1990s
as near-horizon geometries of branes, which is perhaps their most
popular interpretation today. The advent in the mid-1990s of the
``branes and duality'' paradigm led to a renewed effort in the study
of supersymmetric backgrounds. Many such constructions emerged, but by
the end of the decade there was still no systematic approach to the
classification. Since the definition of a supersymmetric background
entails the existence of Killing spinors, which are parallel with
respect to a connection on the spinor bundle, an obvious approach is
via the study of the holonomy of that connection. A first step in that
direction was taken in \cite{JMWaves}, which studied purely
gravitational supersymmetric backgrounds in terms of the possible
lorentzian holonomy groups of eleven-dimensional manifolds admitting
parallel spinors, but it was not clear how to re-introduce the flux in
that approach. Indeed, since the connection with nonzero flux is not
induced from a connection on the spin bundle, there are no theorems
concerning the possible holonomy groups, except that the generic
(restricted) holonomy group is $\SL(32,\RR)$ \cite{HullHolonomy};
although see, e.g.,
\cite{Duff:2003ec,Papadopoulos:2003pf,Bandos:2003us} for some of the
groups that can appear.

One fares a little better starting not from the generic holonomy, but
from the trivial holonomy.  In \cite{FOPMax} the maximally
supersymmetric backgrounds --- i.e., those with trivial (restricted)
holonomy --- were classified, recovering the known maximally
supersymmetric backgrounds: the Freund--Rubin backgrounds
\cite{FreundRubin,AdS7S4} and the gravitational wave of \cite{KG}, in
addition to the trivial Minkowski background.  Attempts to extend this
classification to sub-maximally supersymmetric backgrounds yielded
some negative results: absence of backgrounds with precisely $n=31$
\cite{NoMPreons,FigGadPreons} and $n=30$ \cite{Gran:2010tj}, but the
methods (based on so-called spinorial geometry) become impractical
already for $n=29$.  In fact, we do not even know the size of the
``supersymmetry gap'': the highest known sub-maximal background is a
pp-wave with $n=26$ \cite{Michelson26}, but nothing is known about
$n=27,28,29$.  Methods of spinorial geometry (also confusingly known
as $G$-structures) have also yielded some information at the opposite
end, with local forms of backgrounds for $n=1$ \cite{GauPak,GauGutPak}
in terms of ingredients (such as, Calabi--Yau 5-folds) which offer
little hope of classification.

In this paper we would like to propose a different approach to the
classification, based on the classification of the Killing
superalgebra of the background.  Indeed, every supersymmetric
supergravity background has an associated Lie superalgebra which is
generated by its Killing spinors.  Its construction is reviewed in
Section~\ref{sec:kill-super-supersymm} below.  Its origin is lost in
the mists of time and probably was already understood, at least in
special cases, in the early days of Kaluza--Klein supergravity.  In
more recent times, it made its appearance in the context of the
AdS/CFT correspondence \cite{AFHS,JMFKilling}, brane solutions
\cite{GMT1,GMT2,PKT}, plane waves \cite{FOPFlux,NewIIB} and
homogeneous backgrounds \cite{ALOKilling}, with the general
construction appearing for the first time in \cite{FMPHom} for
$d{=}11$ and \cite{EHJGMHom} for $d{=}10$ supergravities.  Since then,
a number of other supergravity theories have been treated, such as
$d{=}6$ \cite{Figueroa-O'Farrill:2013aca}, $d{=}10$ conformal in
\cite{deMedeiros:2015zrh} and $d{=}4$ (off-shell, minimal) in
\cite{deMedeiros:2016srz}.

The Killing superalgebra has proved to be a very useful invariant of a
supersymmetric supergravity background.  First of all, it
``categorifies'' the fraction of supersymmetry preserved by the
background.  In addition it behaves well under geometric limits, such
as asymptotic and near-horizon limits, but also plane-wave limits.  It
also underlies the (local) homogeneity theorem \cite{FMPHom, EHJGMHom,
  JMF-HC-Lecs, FigueroaO'Farrill:2012fp, Figueroa-O'Farrill:2013aca}
which states that a supergravity background preserving more than half
of the supersymmetry is (locally) homogeneous, which is one of the few
general structural results known about supersymmetric supergravity
backgrounds.

The purpose of this paper is to show that the Killing superalgebra has
a very precise algebraic structure --- one which had passed unnoticed
until recently --- and to derive some of its consequences.  In
particular, we will show that the Killing superalgebra is a filtered
deformation of a $\ZZ$-graded subalgebra of the Poincaré superalgebra.
Let us explain this statement.

Let $(V,\eta)$ denote the lorentzian vector space on which Minkowski
space is modelled, $\fso(V)$ the Lie algebra of the Lorentz group and
$S$ its spinor representation.  In our conventions the inner product
$\eta$ has signature $(1,10)$, i.e., it is ``mostly minus'', and
$S\cong\mathbb R^{32}$ is an irreducible module of the Clifford
algebra $\Cl(V)\cong 2\RR(32)$.  (There are two such modules up to
isomorphism, and they are equivalent as $\fso(V)$-representations.  We
have chosen the module for which the centre acts nontrivially; that
is, for which the action of the volume element $\vol\in\Cl(V)$ is
$\vol\cdot s=-s$ for all $s\in S$.) We recall that $S$ has an
$\fso(V)$-invariant symplectic structure $\left<-,-\right>$ satisfying
\begin{equation*}
  \left<v\cdot s_1, s_2\right> = - \left<s_1, v \cdot s_2\right>,
\end{equation*}
for all $s_1,s_2 \in S$ and $v \in V$, where $\cdot$ refers to the
Clifford action. 

The \emph{Poincaré superalgebra} $\fp$ has underlying vector space
$\fso(V) \oplus S \oplus V$ and nonzero Lie brackets given by the
following expressions, for $A,B \in \fso(V)$, $v \in V$ and $s \in S$:
\begin{equation}
  \label{eq:PSA}
    [A,B] = AB - BA, \qquad [A,s] = \sigma(A)s, \qquad [A,v] = Av, \qquad
    [s,s] = \kappa(s,s).
\end{equation}
Here $\sigma$ is the spinor representation of $\fso(V)$ and
$\kappa(s,s) \in V$ is the \emph{Dirac current} of $s$,
defined by
\begin{equation} 
\label{eq:DiracCurrent}
  \eta(\kappa(s,s),v) = \left<s, v\cdot s\right>,
\end{equation}
for all $v\in V$.  One important property of the Dirac current
$\kappa: \odot^2 S \to V$ is that its restriction to a subspace
$\odot^2 S'$ is still surjective on $V$, provided that the vector
subspace $S'\subset S$ has dimension $\dim S' > 16$.  We shall refer
to this linear algebraic fact as ``local homogeneity'', due to the
crucial rôle it plays in the proof of the local homogeneity theorem of
\cite{FigueroaO'Farrill:2012fp}.

If we grade $\fp$ by declaring $\fso(V)$, $S$ and $V$ to have degrees
$0$, $-1$ and $-2$, respectively, then the above Lie brackets turn
$\fp$ into a ($\ZZ$-)graded Lie superalgebra
\begin{equation*}
  \fp = \fp_{0} \oplus \fp_{-1} \oplus \fp_{-2}, \qquad \fp_0=\fso(V),
  \qquad \fp_{-1}=S,\qquad\fp_{-2}=V.
\end{equation*}
The $\ZZ_2$ grading is compatible with the $\ZZ$ grading, in that
$\fp_{\bar 0} = \fp_0 \oplus \fp_{-2}$ and $\fp_{\bar 1} = \fp_{-1}$;
that is, the parity is the reduction modulo $2$ of the $\ZZ$ degree.

Let now $\fa < \fp$ be a graded subalgebra
that is, $\fa = \fa_0 \oplus \fa_{-1} \oplus \fa_{-2}$, with $\fa_i
\subset \fp_i$.  Recall that a Lie superalgebra $\fg$ is said to be
filtered, if it is admits a vector space filtration
\begin{equation*}
  \fg^\bullet~:\qquad \cdots \supset \fg^{-2} \supset \fg^{-1} \supset
  \fg^0 \supset \cdots,
\end{equation*}
with $\cup_i \fg^i = \fg$ and $\cap_i \fg^i = 0$, which is compatible
with the Lie bracket: that is, $[\fg^i, \fg^j] \subset \fg^{i+j}$.
Associated canonically to every filtered Lie superalgebra
$\fg^\bullet$ there is a graded Lie superalgebra
$\fg_\bullet = \bigoplus_i \fg_i$, where $\fg_i = \fg^i/\fg^{i+1}$.
It follows from the fact that $\fg^\bullet$ is filtered that
$[\fg_i,\fg_j] \subset \fg_{i+j}$, hence $\fg_\bullet$ is graded.  We
say that a Lie superalgebra $\fg$ is a \emph{filtered deformation} of
$\fa < \fp$, if it is filtered and its associated graded superalgebra
is isomorphic (as a graded Lie superalgebra) to $\fa$.  If we do not
wish to mention the subalgebra $\fa$ explicitly, we simply say that
$\fg$ is a \emph{filtered subdeformation} of $\fp$. The first main
result of this paper is the following, which is part of
Theorem~\ref{thm:KSAisFSD}.   That theorem is in turn part of the more
general Theorem~\ref{thm:firstmain} in
Section~\ref{sec:admfilsubdefII}.

\begin{mainthm}
  The Killing superalgebra $\fk=\fk_{\bar 0}\oplus\fk_{\bar 1}$ of an
  $11$-dimensional supergravity background $(M,g,F)$ is a filtered
  subdeformation of the Poincaré superalgebra $\fp$.
\end{mainthm}

Although by ``background'' one typically means a solution of the
(bosonic) field equations, the above result actually only uses the
form of the Killing spinor equation and the fact that
$F \in \Omega^4(M)$ is closed.  A natural question of long standing is
whether some amount of supersymmetry implies the bosonic field
equations.  It is known to be the case for maximal supersymmetry: the
bosonic field equations are equivalent to the vanishing of the
Clifford trace of the gravitino connection, whereas maximal
supersymmetry is equivalent to flatness.  It is also known to fail for
$\leq\tfrac12$-BPS backgrounds, but it has long been suspected that
there is some critical fraction of supersymmetry which forces the
equations of motion.  We give a positive answer to this question in
this paper, where in Section~\ref{sec:field-eqnsII} we prove the
following theorem (see Theorem~\ref{thm:mainII}).

\begin{mainthm}
  Let $(M,g,F)$ be an $11$-dimensional lorentzian spin manifold endowed
  with a closed $4$-form $F\in\Omega^4(M)$.  If the real vector space
  \begin{equation*}
    \fk_{\bar 1} = \left\{ \varepsilon \in \Gamma(\$) ~\middle |~
      \nabla_X \varepsilon = \tfrac1{24} (X \cdot F - 3 F \cdot X)
      \cdot \varepsilon\right\}
  \end{equation*}
  of Killing spinors has dimension $\dim\fk_{\bar 1} > 16$, then
  $(M,g,F)$ satisfies the bosonic field equations of $11$-dimensional
  supergravity.
\end{mainthm}

The above condition on the dimension of the space of Killing spinors is
crucial to many of our results and we have tentatively given it the
name of ``high supersymmetry''.  We will therefore refer to ``highly
supersymmetric backgrounds'' when talking about backgrounds preserving
more than half of the supersymmetry.  Similarly, we will say that a
filtered subdeformation $\fg=\fg_{\bar 0} \oplus \fg_{\bar 1}$ is ``highly
supersymmetric'' if $\dim\fg_{\bar 1} > 16$.

The two theorems quoted above suggest an approach to the
classification of highly supersymmetric backgrounds via the
classification of their Killing superalgebras, which as mentioned
above are (certain) filtered subdeformations of the Poincaré
superalgebra.  A first step in such a research programme was taken in
\cite{Figueroa-O'Farrill:2015efc,Figueroa-O'Farrill:2015utu}, where we
recovered the classification in \cite{FOPMax} of maximally
supersymmetric supergravity backgrounds. We also refer to the
introduction in \cite{Figueroa-O'Farrill:2015efc} and to
\cite{MR3056953,MR2798219,MR3218266,MR3255456} for more details on the
underlying geometric interpretation of the Killing superalgebra in the
context of ``nonholonomic'' $G$-structures on supermanifolds.

Of course, there is no reason to believe that \emph{any} filtered
subdeformation of the Poincaré superalgebra is the Killing
superalgebra of a supersymmetric background and one of the aims of
this paper is to characterise algebraically those filtered
subdeformations which are Killing superalgebras of highly
supersymmetric backgrounds.  It would be interesting to characterise
the filtered deformations which are Killing superalgebras of
supersymmetric backgrounds, regardless the amount of supersymmetry
preserved, but we don't do that in this paper.

Thus we will narrow down the class of filtered subdeformations
$\fg = \fg_{\bar 0} \oplus \fg_{\bar 1}$ to those which are highly
supersymmetric ($\dim\fg_{\bar 1} > 16$) and which satisfy additional
criteria set out in Definition~\ref{def:realizable}.  These criteria
essentially amount to demanding that $\fg$ should be constructed out
of a closed $4$-form in a way consistent with supergravity.  We will
say that those highly supersymmetric filtered subdeformations are
(geometrically) realizable.

The Killing superalgebra of a highly supersymmetric background is
realizable (see Theorem \ref{thm:KSAisFSD}).  Conversely, if a highly
supersymmetric filtered subdeformation $\fg$ is realizable, it is then
possible to reconstruct a background whose Killing superalgebra
contains $\fg$.  This result is contained in a partial converse of the
first quoted theorem above, which also forms part of
Theorem~\ref{thm:firstmain}.

\begin{mainthm}
  Any realizable highly supersymmetric filtered subdeformation $\fg$
  of the Poincaré superalgebra $\fp$ is a subalgebra of the Killing
  superalgebra of a highly supersymmetric $11$-dimensional
  supergravity background.
\end{mainthm}

We make no claims about the existence of a similar result when we drop
the ``highly supersymmetric'' condition: our reconstruction theorem
requires (local) homogeneity, which only high supersymmetry
guarantees.

As a consequence of these results we will be able to map the
classification problem of highly supersymmetric backgrounds of
$11$-dimensional supergravity to the classification problem of
(maximal) realizable highly supersymmetric filtered subdeformations of
the Poincaré superalgebra.  A refined version of this approach, where
we restrict to the classification of Killing ideals (see below),
corresponding to classifying realizable highly supersymmetric filtered
subdeformations which are odd-generated, seems slightly more
tractable.

This paper is organised as follows.  In Section~\ref{sec:filt-def} we
summarise the basic notions and results about filtered deformations of
Lie superalgebras and in particular of the Poincaré superalgebra.  We
also define the notion of a (highly supersymmetric) realizable
filtered subdeformation, since those are the ones which can correspond
to Killing superalgebras of highly supersymmetric $11$-dimensional
supergravity backgrounds.  In Section~\ref{sec:kill-super-supersymm}
we review the geometric construction of the Killing superalgebra and
in Section~\ref{sec:admfilsubdef} we prove that the Killing
superalgebra is a filtered subdeformation of $\fp$.  In
Section~\ref{sec:admfilsubdefII}, we prove our first main result:
Theorem~\ref{thm:firstmain}.  In Section~\ref{sec:class-probl-kill} we
consider the Jacobi identity of the Killing superalgebra and define
the classification problem for highly supersymmetric Killing
superalgebras as the classification of realizable highly supersymmetric
filtered subdeformations.  We then describe the latter in
terms of simpler objects.  In Section~\ref{sec:field-eqns} we relate
the Jacobi identity to the supergravity field equations. We first
recall some useful algebraic and differential identities from
\cite{GauPak,GauGutPak}, and then show that high supersymmetry implies
the field equations (see Theorem~\ref{thm:mainII}). We conclude with
some observations.

\section{Preliminaries on filtered deformations}
\label{sec:filt-def}

Let $\fp$ be the Poincaré superalgebra.  In this section we discuss
filtered deformations of its $\ZZ$-graded subalgebras.

\subsection{Basic definitions and results}

Let $\fa=\fa_0\oplus\fa_{-1}\oplus\fa_{-2}$ be a $\mathbb Z$-graded
subalgebra of $\fp$, where $\fa_{-2}=V'\subset V$,
$\fa_{-1}=S'\subset S$ and $\fa_0=\fh$ is a Lie subalgebra of
$\fso(V)$. We denote the negatively graded part of $\fa$ by
$\fa_-=\fa_{-1}\oplus\fa_{-2}$; in particular
$\fp_-= \fp_{-1} \oplus \fp_{-2}$, $\fp_{-2}=V$ and $\fp_{-1}=S$, is
the usual (2-step nilpotent) supertranslation ideal of the Poincaré
superalgebra.

\begin{definition}
  The subalgebra $\fa$ of $\fp$ is called \emph{highly supersymmetric}
  if $\dim S'>16$.
\end{definition}

We recall that a $\mathbb Z$-graded Lie superalgebra $\fa=\oplus\fa_p$
with negatively graded part $\fa_-=\oplus_{p<0}\fa_p$ is called
\emph{fundamental} if $\fa_-$ is generated by $\fa_{-1}$ and
\emph{transitive} if for any $x\in\fa_p$ with $p\geq 0$ the condition
$[x,\fa_-]=0$ implies $x=0$. It is not hard to exhibit graded
subalgebras $\fa$ of $\fp$ for which $S'$ has dimension $16$ and $V'$
is a proper subspace of $V$.  On the other hand, we have the following

\begin{lemma}
  \label{lem:fil1}
  Let $\fa$ be a highly supersymmetric graded subalgebra of
  $\fp$. Then $\fa_{-2}=V$ and $\fa$ is fundamental and transitive.
\end{lemma}

\begin{proof}
  The algebraic fact underlying the local homogeneity theorem in
  \cite{FigueroaO'Farrill:2012fp} says precisely that the image of $\kappa$
  restricted to $S'\otimes S'$ equals $V$ if $\dim S'>16$ . It follows
  that $\fa_{-2}=V$ and $\fa$ is fundamental.  The transitivity of
  $\fa$ follows from the fact that $V$ is a faithful representation of
  any Lie subalgebra of $\fso(V)$.
\end{proof}

To proceed further, we first need to recall the definition of an
appropriate complex associated with $\fa$. It is called the
(generalised) Spencer complex and it is a refinement (by degree) of
the usual Chevalley--Eilenberg complex of a Lie superalgebra.

\begin{table}[h!]
  \centering
  \caption{Even $q$-cochains of small degree}
  \label{tab:even-cochains-small}
  \setlength\extrarowheight{3pt}
  \begin{tabular}{c|*{5}{>{$}c<{$}}}
    \multicolumn{1}{c|}{} & \multicolumn{5}{c}{$q$} \\
    deg & 0 & 1 & 2 & 3 & 4 \\\hline
    0 & \fh & \begin{tabular}{@{}>{$}c<{$}@{}} S' \to S'\\ V' \to
                    V'\end{tabular} & \odot^2 S' \to V' & & \\\hline
    2 & & V' \to \fh & \begin{tabular}{@{}>{$}c<{$}@{}}
                            \Lambda^2 V' \to V'\\ V' \otimes S' \to S' \\
                            \odot^2 S' \to \fh \end{tabular}
         & \begin{tabular}{@{}>{$}c<{$}@{}} \odot^3 S' \to S'\\
             \odot^2 S' \otimes V' \to V'\end{tabular} & \odot^4 S' \to V'
        \\\hline
    4 & & & \Lambda^2 V' \to \fh & \begin{tabular}{@{}>{$}c<{$}@{}}
                                        \odot^2 S' \otimes V' \to
                                        \fh \\ \Lambda^2 V' \otimes
                                        S' \to S' \\ \Lambda^3 V' \to
                                        V' \end{tabular} & \begin{tabular}{@{}>{$}c<{$}@{}}
                                        \odot^4 S' \to \fh \\
                                        \odot^3 S' \otimes V' \to S' \end{tabular} \\
  \end{tabular}
\end{table}

The cochains of the Spencer complex of $\fa$ are linear maps
$\Lambda^q \fa_-\to \fa$ or, equivalently, elements of
$\fa \otimes \Lambda^q \fa_-^*$, where $\Lambda^ \bullet$ is meant
here in the super sense. We extend the degree in $\fa$ to such
cochains by declaring that $\fa_p^*$ has degree $-p$. The spaces in
the complexes of even cochains for small degree are given in
Table~\ref{tab:even-cochains-small}; although for degree $d=4$ there
are cochains also for $q=5,6$ which we omit.

Let $C^{d,q}(\fa_-,\fa)$ be the space of $q$-cochains of degree $d$.
The Spencer differential
\begin{equation*}
  \partial: C^{d,q}(\fa_-,\fa) \to C^{d,q+1}(\fa_-,\fa)
\end{equation*}
coincides with the restriction of the Chevalley--Eilenberg
differential for the Lie superalgebra $\fa_-$ relative to its module
$\fa$ with respect to the adjoint action.

For $q=0,1,2$ and $d\equiv 0 \pmod 2$, the Spencer differential is
explicitly given by the following expressions:
\begin{align}
  \partial\phi(x) &= [x,\phi]   \label{eq:Spencer0}\\
  \partial\phi(x,y) &= [x,\phi(y)] - (-1)^{|x||y|} [y,\phi(x)] - \phi([x,y]) \label{eq:Spencer1}\\
  \partial\phi(x,y,z) &= [x,\phi(y,z)]+(-1)^{|x|(|y|+|z|)}[y,\phi(z,x)] + (-1)^{|z|(|x|+|y|)} [z,\phi(x,y)]\notag \\
    & \quad {} - \phi([x,y],z) - (-1)^{|x|(|y|+|z|)} \phi([y,z],x) -(-1)^{|z|(|x|+|y|)} \phi([z,x],y),\label{eq:Spencer2}
\end{align}
where $|x|,|y|,\dots$ are the parity of elements $x,y,\dots$ of $\fa_-$
and $\phi\in C^{d,q}(\fa_-,\fa)$ with $q=0,1,2$, respectively.

We say that a $\ZZ$-graded Lie superalgebra $\fa$ with
negatively graded part $\fa_-$ is a \emph{full prolongation of degree
  $k$} if $H^{d,1}(\fa_-,\fa)=0$ for all $d\geq k$ (see \cite{MR1688484}).

\begin{lemma}
  \label{lem:fil2}
  Let $\fa$ be a highly supersymmetric graded subalgebra of
  $\fp$. Then $\fa$  is a full prolongation of degree $2$ and
  $H^{d,2}(\fa_-,\fa)=0$ for all even $d\geq 4$.
\end{lemma}

\begin{proof}
  Since $\fa$ is fundamental, many of the components of the Spencer
  differential turn out to be injective. For instance every
  $\phi\in C^{2,1}(\fa_-,\fa)=\Hom(V,\fh)$ satisfies
  $\partial\phi(s_1,s_2)=-\phi(\kappa(s_1,s_2))$ for all
  $s_1,s_2\in S'$ so that $\phi=0$ is the only cocycle and
  $H^{2,1}(\fa_-,\fa)=0$. If $d>2$ the space of cochains
  $C^{d,1}(\fa_-,\fa)=0$ and first claim follows. 
  
  If $\phi\in C^{4,2}(\fa_-,\fa)=\Hom(\Lambda^2 V, \fh)$, one has
  $\partial\phi(s_1,s_2,v)=-\phi(\kappa(s_1,s_2),v)$ where
  $s_1,s_2\in S'$ and $v\in V$. In particular
  $\ker \partial|_{C^{4,2}(\fa_-,\fa)}=0$ and
  $H^{4,2}(\fa_-,\fa)=0$. If $d>4$ then $C^{d,2}(\fa_-,\fa)=0$ and
  last claim follows.
\end{proof}

Let $\fg$ be a filtered deformation of a graded subalgebra
$\fa=\fh\oplus S'\oplus V'$ of $\fp$. It satisfies
$[\fg^i, \fg^j] \subset \fg^{i+j}$, where the filtration $\fg^\bullet$
is
\begin{equation*}
  \fg^\bullet : \qquad \fg = \fg^{-2} \supset \fg^{-1} \supset \fg^0
  \supset 0,\qquad\fg^{-1} = \fh \oplus S',\qquad\fg^0 = \fh,
\end{equation*}
hence its Lie brackets take the following general form 
\begin{equation}
\label{eq:generalbrackets}
  \begin{aligned}[m]
    [A,B]&=AB-BA\\
    [A,s]&=\sigma(A)s\\
    [A,v]&=Av+\delta(A,v)
  \end{aligned}
  \qquad\qquad
  \begin{aligned}[m]
    [s,s]&=\kappa(s,s)+\gamma(s,s)\\
    [v,s]&=\beta(v,s)\\
    [v,w]&=\alpha(v,w)+\rho(v,w),
  \end{aligned}
\end{equation}
for some maps $\alpha\in\Hom(\Lambda^2 V',V')$,
$\beta\in\Hom(V'\otimes S',S')$, $\gamma\in\Hom(\odot^2 S',\fh)$ and
also $\delta\in\Hom(\fh\otimes V',\fh)$ of degree $2$ and a map
$\rho\in\Hom(\wedge^2V',\fh)$ of degree $4$, where $A,B\in \fh$,
$s\in S'$ and $v,w\in V'$.

\begin{definition}\label{def:highsusy}
  The filtered deformation $\fg$ of $\fa$ is called \emph{highly
    supersymmetric} if $\fa$ is highly supersymmetric; that is, if
  $\dim\fg_{\bar 1}=\dim S'>16$.
\end{definition}

To introduce the notion of isomorphism between filtered
subdeformations of $\fp$, we note that the spin group $\Spin(V)$
naturally acts on $\fp=\fso(V)\oplus S\oplus V$ by $0$-degree Lie
superalgebra automorphisms. In particular any element $g\in\Spin(V)$
sends a graded subalgebra of $\fp$ into an (isomorphic) graded
subalgebra of $\fp$.

\begin{definition} 
\label{def:isofil}
 An \emph{isomorphism}  of filtered subdeformations $\fg$ and
  $\widetilde\fg$ of $\fp$ is a map
  $\Phi:\fg\longrightarrow\widetilde\fg$ such that:
  \begin{enumerate}[label=(\roman*)]
  \item $\Phi$ is an isomorphism of Lie superalgebras; 
  \item $\Phi$ is compatible with the filtrations;
    i.e., $\Phi(\fg^i)=\widetilde\fg^i$ for $i=-2,-1,0$;
  \item the induced $0$-degree Lie superalgebra isomorphism of
    associated graded Lie superalgebras $\fa$ and $\widetilde \fa$ is
    given by the natural action of some $g\in\Spin(V)$.\footnote{It is
      well known that the bosonic field equations of $11$-dimensional
      supergravity are invariant under a homothety which rescales both
      the metric and the 4-form.  With this definition of isomorphism
      of filtered subdeformations, the Killing superalgebras of two
      different supergravity backgrounds related by a homothety would
      \emph{not} be isomorphic.  They would be isomorphic, however,
      were we to modify the definition by replacing $\Spin(V)$ by
      $\CSpin(V)$ in (iii).}
  \end{enumerate}
  If we do not wish to mention $\Phi$ explicitly, we simply say that
  $\fg$ and $\widetilde\fg$ are \emph{isomorphic}.
\end{definition}

It is easy to see that an isomorphism $\Phi:\fg\to\widetilde\fg$ takes
the following general form, for some $g\in\Spin(V)$ and
$X':V'\to\widetilde\fh$:
\begin{equation}
  \label{eq:generaliso}
  \Phi(A)=g\cdot A,\qquad
  \Phi(s)=g\cdot s,\qquad\text{and}\qquad
  \Phi(v)=g\cdot v+ X'_v,
\end{equation}
where $A\in\fh$, $s\in S'$ and $v\in V'$. In the following, we
consider isomorphisms of highly supersymmetric filtered
subdeformations whose associated $0$-degree map is the identity, that
is with $g=e$ in \eqref{eq:generaliso}. We denote the sum of all
components in \eqref{eq:generalbrackets} of degree $2$ by the symbol
$\mu=\alpha+\beta+\gamma+\delta:\fa\otimes\fa\to\fa$, where
$\alpha,\beta,\gamma,\delta$ are the maps introduced just before
Definition~\ref{def:highsusy}.

\begin{proposition}
  \label{prop:fil1}
  Let $\fg$ be a highly supersymmetric filtered deformation of a
  graded subalgebra $\fa$ of $\fp$.  Then:
  \begin{enumerate}
  \item $\mu|_{\fa_-\otimes\fa_-}$ is a cocycle in $C^{2,2}(\fa_-,\fa)$ and
    its cohomology class
    \begin{equation*}
      [\mu|_{\fa_-\otimes\fa_-}]\in H^{2,2}(\fa_-,\fa)
    \end{equation*}
    is $\fh$-invariant (that is, the cocycle $\mu|_{\fa_-\otimes\fa_-}$ is
    $\fh$-invariant up to coboundaries);
  \item if $\widetilde\fg$ is another filtered deformation of the same
    $\fa$ such that $[\widetilde\mu|_{\fa_-\otimes\fa_-}]=
    [\mu|_{\fa_-\otimes\fa_-}]$ then $\widetilde\fg$ is isomorphic to
    $\fg$.
  \end{enumerate}
\end{proposition}

\begin{proof}
  The proof relies on general results on filtered deformations in
  \cite{MR1688484} and the full line of arguments is the same as in
  \cite[Theorem 9]{Figueroa-O'Farrill:2015efc} or \cite[Proposition
  10]{deMedeiros:2016srz}. Here we simply record that, since $V'=V$,
  any $\alpha\in\Hom(\Lambda^2 V,V)$ can be written as
  \begin{equation*}
    \alpha(v,w)=X_{v}w-X_{w}v,\qquad v,w\in V,
  \end{equation*}
  for a unique linear map $X:V\to\fso(V)$ and that 
  \begin{equation*}
    \begin{split}
      \delta(A,v)&=[A,X_v]-X_{Av},\\
      \partial
      \rho(s,s,v)&=2\gamma(s,\beta(v,s))+\delta(\gamma(s,s),v),
    \end{split}
  \end{equation*}
  for all $A\in\fh$, $v\in V$ and $s\in S'$. It follows from
  Lemmas~\ref{lem:fil1}~and~\ref{lem:fil2} that the components
  $\delta$ and $\rho$ are uniquely determined, once
  $\mu|_{\fa_-\otimes\fa_-}=\alpha+\beta+\gamma$ has been fixed. In a
  similar way the components $\widetilde\delta$ and $\widetilde\rho$
  of $\widetilde\fg$ are also fixed in terms of
  $\widetilde\mu|_{\fa_-\otimes\fa_-}$.

  By hypothesis $\widetilde\mu|_{\fa_-\otimes\fa_-} =
  \mu|_{\fa_-\otimes\fa_-} - \partial X'$, with $X':V\to\fh$ giving
  the required isomorphism.
\end{proof}

We have seen that highly supersymmetric filtered subdeformations with
associated graded subalgebra $\fa$ of $\fp$ are determined, up to
isomorphisms of filtered subdeformations, by the space
$H^{2,2}(\fa_-,\fa)^{\fh}$ of $\fh$-invariant elements in
$H^{2,2}(\fa_-,\fa)$.  We will obtain improved versions of Proposition
\ref{prop:fil1} in Section~\ref{sec:class-probl-kill}, in the case of
(odd-generated) realizable filtered subdeformations. The concept of
realizability is introduced in the next section.

\subsection{Realizable filtered deformations}

In \cite[Proposition~7]{Figueroa-O'Farrill:2015efc}, we determined the
group $H^{2,2}(\fa_-,\fa)$ when $\fa=\fp$, and found that
$H^{2,2}(\fp_-,\fp)\cong\wedge^4 V$ as an $\fso(V)$-module. More
precisely any class admits a canonical representative of the form
$\beta^{\varphi}+\gamma^{\varphi}$ for a unique
$\varphi\in\Lambda^4 V$, where $\beta^\varphi:V\otimes S\to S$ and
$\gamma^\varphi:\odot^2 S\to\fso(V)$ are given by
\begin{equation}
\label{eq:thetatoo}
\begin{aligned}
\beta^{\varphi}(v,s)&=\tfrac1{24}(v\cdot\varphi-3\varphi\cdot v)\cdot s,\\
\gamma^{\varphi}(s,s)(v)&=-2\kappa(\beta^\varphi(v,s),s),
\end{aligned}
\end{equation}
for all $s\in S$ and $v\in V$.

Associated to the natural inclusion $\imath:\fa\to\fp$ there are chain
maps $\imath_*:C^{\bullet,\bullet}(\fa_-,\fa)\to
C^{\bullet,\bullet}(\fa_-,\fp)$ and $\imath^*:C^{\bullet,\bullet}(\fp_-,\fp)\to
C^{\bullet,\bullet}(\fa_-,\fp)$ inducing the corresponding maps in
cohomology
\begin{equation}
  \label{eq:mapadm}
  \begin{split}
    \imath_*&:H^{2,2}(\fa_-,\fa)\to H^{2,2}(\fa_-,\fp),\\
    \imath^*&:H^{2,2}(\fp_-,\fp)\to H^{2,2}(\fa_-,\fp).
  \end{split}
\end{equation}
Both maps in \eqref{eq:mapadm} are $\fh$-equivariant. Moreover we have
the following

\begin{lemma}
  \label{lem:kernelpriori}
  Let $\fa$ be a highly supersymmetric graded subalgebra of $\fp$.
  Then $\imath_*$ is injective and $\ker
  \imath^*\cong\left\{\varphi\in\Lambda^4 V ~ \middle|
    ~\beta^\varphi|_{V\otimes  S'}=0\right\}$.
\end{lemma}

\begin{proof}
  Let $[\alpha+\beta+\gamma]\in H^{2,2}(\fa_-,\fa)$ be such that
  $\imath_*[\alpha+\beta+\gamma]=0$; that is, $\alpha + \beta + \gamma
  = \partial\phi$ for some $\phi : V \to \fso(V)$.  However, in
  particular $\phi(\kappa(s_1,s_2))=-\gamma(s_1,s_2)\in\fh$ for all
  $s_1,s_2\in S'$, hence $\phi(v)\in\fh$ for all $v\in V$ and
  $[\alpha+\beta+\gamma]=0$.  This proves the first claim. The second
  claim is straightforward.
\end{proof}

\begin{remark}
  We will see that the space $\ker \imath^*$ parametrises
  the $4$-forms compatible with a highly supersymmetric flat
  supergravity background, so that $\imath^*$ too is injective (see
  Corollary \ref{cor:final!}).  It would be desirable to have an a
  priori representation-theoretic proof of this fact.
\end{remark}

In the following definition we introduce the concept of realizability
for highly supersymmetric filtered subdeformations of $\fp$.  This is
an algebraic criterion which, as we will see, singles out the filtered subdeformations which are
geometrically realizable as (subalgebras of) Killing superalgebras of supergravity backgrounds.  Since the main
purpose of this paper is to lay the foundations for the classification
of highly supersymmetric backgrounds, we restrict the definition to
the highly supersymmetric case.  This is not to suggest that the
notion of a (geometrically) realizable filtered subdeformation is not
interesting if $\dim\fg_{\bar 1}\leq 16$, but the definition
would be more involved.

\begin{definition}
  \label{def:realizable}
  A filtered deformation $\fg$ of a highly supersymmetric $\mathbb
  Z$-graded subalgebra $\fa$ of $\fp$ is said to be \emph{realizable}
  if the following conditions are satisfied:
  \begin{enumerate}[label=(\roman*)]
  \item the associated cohomology class $[\alpha+\beta+\gamma]\in
    H^{2,2}(\fa_-,\fa)$ is of the form
    \begin{equation}
      \label{eq:WAdm}
      \imath_*[\alpha+\beta+\gamma]=\imath^*[\beta^\varphi+\gamma^{\varphi}]
    \end{equation}
    for some $\varphi\in\Lambda^4 V$ (the collection of such forms is
    an affine space modeled on the vector space $\ker \imath^*$); and
  \item there exists a $\varphi\in\Lambda^4 V\cong\Lambda^4 V^*$ as in
    (i) which is also $\fh$-invariant and closed. The condition for
    an $\fh$-invariant $\varphi$ to be closed is equivalent to
    \begin{equation}
      \label{eq:Adm}
      d\varphi(v_0,\ldots,v_4) = \sum_{i<j}(-1)^{i+j}
      \varphi(\alpha(v_i,v_j), v_0, \ldots, \hat v_i,\ldots,\hat
      v_j,\ldots,v_4) = 0
    \end{equation}
    for all $v_0,\ldots, v_4\in V$. 
    \end{enumerate}
  We call \emph{admissible} any $\varphi\in\Lambda^4 V$ which is
  $\fh$-invariant and satisfies \eqref{eq:WAdm} and \eqref{eq:Adm}.
\end{definition}

\begin{definition}
  \label{def:srealizable}
  A filtered deformation $\fg$ is called \emph{odd-generated realizable} if
  it is realizable and generated by the odd part, that is $\fg_{\bar
    0}=[\fg_{\bar 1},\fg_{\bar 1}]$.
\end{definition}

We will see below that the Killing spinors of a highly supersymmetric
background generate a filtered Lie superalgebra which is
odd-generated and realizable.

We note that the condition of being (odd-generated) realizable is
preserved by the isomorphisms of filtered subdeformations of $\fp$. In
particular, it is worth remarking that even though the
condition~\eqref{eq:Adm} that $\varphi$ be closed seems to depend
explicitly on $\alpha$, it actually only depends on the cohomology
class of $\alpha + \beta + \gamma$ in $H^{2,2}(\fa_-,\fa)$. Indeed, if
we modify the cocycle by a coboundary $\partial\phi$, for some
$\phi: V \to \fh$, then $\alpha$ changes to
$\widetilde\alpha(v,w) = \alpha(v,w) + \phi_v w - \phi_w v$, and using
the $\fh$-invariance of $\varphi$ one sees that the expression for
$d\varphi$ in equation~\eqref{eq:Adm} remains unchanged.

It follows from \eqref{eq:WAdm} that the Lie brackets of a realizable
filtered deformation $\fg$ of $\fa=\fh\oplus S'\oplus V$ are as in
\eqref{eq:generalbrackets}, where $V'=V$ and
\begin{equation}
  \label{eq:generalbracketsII}
  \begin{split}
	\alpha(v,w) &= X_v w - X_w v\\
	\beta(v,s)&=\beta^\varphi(v,s) + \sigma(X_v) s\\
	\gamma(s,s)&=\gamma^\varphi(s,s) - X_{\kappa(s,s)}\\
    \delta(A,v) &= [A,X_v] - X_{Av},\\    
		\end{split}
\end{equation}
for some linear map $X:V\to\fso(V)$. Here $A,B \in \fh$, $v,w \in V$,
$s \in S'$. This implies that any highly supersymmetric graded
subalgebra $\overline\fa=\overline\fh\oplus\overline S'\oplus V$ of $\fa$
which is also closed under the Lie brackets of $\fg$ inherits a
natural structure $\overline\fg$ of realizable filtered
subdeformation. This motivates the following

\begin{definition}
  An \emph{embedding} of filtered subdeformations of $\fp$
  is an injective map $\Phi:\widetilde\fg\longrightarrow\fg$ such
  that $\Phi : \widetilde\fg \longrightarrow \overline\fg =
  \Phi(\widetilde\fg)$ is an isomorphism of filtered subdeformations,
  where $\overline\fg \subseteq \fg$ has the natural structure of
  filtered subdeformation induced by $\fg$.
 \end{definition}

\section{The Killing superalgebra as a filtered deformation}
\label{sec:kill-super-as}

We will first review in Section~\ref{sec:kill-super-supersymm} the
construction of the Killing superalgebra
$\fk=\fk_{\bar 0}\oplus\fk_{\bar 1}$ associated to a supersymmetric
background $(M,g,F)$ of $11$-dimensional supergravity, following the
description in \cite{FMPHom}.  Actually, it is known that $\fk$ can be
constructed for any $11$-dimensional lorentzian spin manifold
$(M,g,F)$ endowed with a closed $4$-form $F\in\Omega^4(M)$. In other
words, the supergravity Einstein and Maxwell equations are of no
consequence in what follows in Section~\ref{sec:kill-super-as}.

We will then show in Section~\ref{sec:admfilsubdef} that the Killing
superalgebra, as a filtered Lie superalgebra, is (isomorphic to) a
filtered subdeformation of the Poincaré superalgebra. The main result
of this section deals with the highly supersymmetric case and it is
given by Theorem \ref{thm:firstmain} in
Section~\ref{sec:admfilsubdefII}.

\subsection{The Killing superalgebra}
\label{sec:kill-super-supersymm}

Let $(M,g,F)$ be a connected $11$-dimensional lorentzian spin manifold
endowed with a closed $4$-form $F\in\Omega^4(M)$.  We denote the
Levi-Civita connection by $\nabla$ and the associated spinor bundle by
$\$ \to M$ (more precisely, this is a bundle of Clifford modules over
$\Cl(TM)$ associated to one of the two non-isomorphic irreducible
Clifford modules).

The spinor fields
$\varepsilon \in \Gamma(\$)$ which satisfy, for all vector fields $Z
\in \eX(M)$,
\begin{equation*}
  \nabla_Z \varepsilon = \tfrac1{24} (Z \cdot F - 3 F \cdot Z) \cdot
  \varepsilon~
\end{equation*}
are called \emph{Killing spinors} and they define a real vector space
$\fk_{\bar 1}$. We also let $\fk_{\bar 0}$ be the space of $F$-preserving
Killing vectors.  

The Killing superalgebra is a Lie superalgebra
structure on $\fk = \fk_{\bar 0} \oplus
\fk_{\bar 1}$; we will review the construction below.
For our purposes, it is convenient to introduce bundle morphisms
$\beta^F:TM\otimes\$\to\$$ and $\gamma^F:\odot^2 \$\to\fso(TM)$
defined by
\begin{equation}
  \label{eq:theta}
  \begin{split}
    \beta^{F}(Z,\varepsilon)&=\tfrac1{24} (Z\cdot F - 3 F \cdot Z) \cdot \varepsilon,\\
    \gamma^{F}(\varepsilon,\varepsilon)(Z)&=-2\kappa(\beta^F(Z,\varepsilon),\varepsilon),
  \end{split}
\end{equation}
where $Z\in \eX(M)$ and
$\varepsilon\in\Gamma(\$)$. In particular Killing spinors are exactly
those spinors $\varepsilon$ which satisfy $\nabla_Z
\varepsilon=\beta^F(Z,\varepsilon)$ for all $Z\in \eX(M)$.

Let $\eE = \eE_{\bar 0} \oplus \eE_{\bar 1}$ be the supervector bundle where
$\eE_{\bar 0} = TM \oplus \fso(TM)$ and $\eE_{\bar 1} = \$$.  On $\eE$ we
have an even connection $D$ defined by
\begin{equation*}
  D_Z \varepsilon = \nabla_Z \varepsilon - \beta^F(Z,\varepsilon),
\end{equation*}
for $\varepsilon \in \Gamma(\eE_{\bar 1})$ and
\begin{equation*}
  D_Z
  \begin{pmatrix}
    \xi \\ \Xi
  \end{pmatrix}
  =
  \begin{pmatrix}
    \nabla_Z\xi + \Xi(Z)\\ \nabla_Z\Xi - R(Z,\xi)
  \end{pmatrix},
\end{equation*}
for $(\xi,\Xi) \in \Gamma(\eE_{\bar 0})$, where $R:\Lambda^2 TM \to
\fso(TM)$ is the Riemann curvature.  Then
\begin{equation*}
  \begin{split}
    \fk_{\bar 1} &= \left\{\varepsilon \in \Gamma(\eE_{\bar 1})
      ~ \middle | ~ D\varepsilon=0\right\}\\
    \fk_{\bar 0} &= \left\{(\xi,\Xi) \in \Gamma(\eE_{\bar 0}) ~
      \middle | ~ D (\xi,\Xi) =0 \quad\text{and}\quad \nabla_\xi F +
      \Xi\cdot F = 0\right\},
  \end{split}
\end{equation*}
where $\Xi\cdot F$ is the natural action of $\fso(TM)$ on $4$-forms.  In
particular, an element of the Killing superalgebra is determined by
the value at a point in $M$ of the corresponding parallel
section of $\eE$.  In other words, given any point $o \in M$, the
Killing superalgebra defines a vector subspace of $\fso(T_o M) \oplus
\$_o \oplus T_o M$.  

We will introduce the notation
\begin{equation*}
    (V,\eta) = (T_o M, g_o), \qquad \fso(V) = \fso(T_o M), \qquad
    S = \$_o ,
\end{equation*}
so that $(V,\eta)$ is an
$11$-dimensional lorentzian vector space with Lie algebra
$\fso(V)$ of skew-symmetric endomorphisms, and
$S$ an irreducible $\Cl(V)$-module. Notice that $\fso(V) \oplus S
\oplus V$ is the vector space underlying the Poincaré superalgebra.

We now describe the Lie brackets of $\fk$. Let
$(\xi,X_\xi), (\zeta,X_\zeta) \in \fk_{\bar 0}$.  This means that
$\xi,\zeta$ are $F$-preserving Killing vector fields with
$X_\xi = - \nabla \xi$ and $X_\zeta = - \nabla \zeta$.  Their Lie
bracket is given by
\begin{equation}
  \label{eq:LiebracketR}
  [(\xi,X_\xi),(\zeta,X_\zeta)] = (X_\xi \zeta - X_\zeta \xi,
  [X_\xi,X_\zeta] + R(\xi,\zeta)),
\end{equation}
with the Riemann curvature measuring the deviation of $\fk_{\bar 0}$
from being a subalgebra of the Poincaré algebra $\fp_{\bar 0}$.  Now
let $\varepsilon \in \fk_{\bar 1}$ be a Killing spinor.  The action of
$\fk_{\bar 0}$ on $\fk_{\bar 1}$ is given by the spinorial Lie
derivative (\cite{MR0312413}; see also, e.g., \cite{JMFKilling})
\begin{equation}
  \label{eq:spinLieDer}
  \eL_\xi \varepsilon = \nabla_\xi \varepsilon + \sigma(X_\xi)
  \varepsilon,
\end{equation}
where $\sigma$ is the spinor representation of $\fso(TM)$.  From the
fact that $D \varepsilon = 0$, we may rewrite this action without
derivatives:
\begin{equation}
  \label{eq:LieBracketBeta}
  [(\xi,X_\xi), \varepsilon] = \beta^F(\xi,\varepsilon) +
  \sigma(X_\xi) \varepsilon.
\end{equation}
Lastly, the square of a Killing spinor is its Dirac current, which
belongs to $\fk_{\bar 0}$ (\cite{FMPHom}; see also, e.g.,
Corollary~\ref{cor:hsMax} in Section~\ref{sec:field-eqnsI}):
\begin{equation*}
  [\varepsilon,\varepsilon] = (\kappa(\varepsilon,\varepsilon) ,
  -\nabla\kappa(\varepsilon,\varepsilon)).
\end{equation*}
A calculation shows that
\begin{equation}
  \begin{split}
    \label{eq:LieBracketSS}
    -\nabla\kappa(\varepsilon,\varepsilon)(Z) &= - \nabla_Z
    \kappa(\varepsilon, \varepsilon) = - 2 \kappa(\nabla_Z \varepsilon,
    \varepsilon)\\ 
    &= - \tfrac1{12} \kappa\left((Z \cdot F - 3 F \cdot
      Z) \cdot \varepsilon, \varepsilon\right)\\
    &=\gamma^F(\varepsilon,\varepsilon)(Z),
  \end{split}
\end{equation}
for all vector fields $Z\in \eX(M)$.

In summary, the Lie brackets of $\fk$ are given, for
$(\xi,X_\xi), (\zeta,X_\zeta) \in \fk_{\bar 0}$ and
$\varepsilon \in \fk_{\bar 1}$, by the following:
\begin{equation}
  \label{eq:KSA}
  \begin{split}
    [(\xi,X_\xi), (\zeta,X_\zeta)] &= (X_\xi\zeta-X_\zeta\xi,
    [X_\xi,X_\zeta] + R(\xi,\zeta))\\
    [(\xi,X_\xi), \varepsilon] &= \beta^{F}(\xi,\varepsilon) + \sigma(X_\xi)\varepsilon\\
    [\varepsilon, \varepsilon] &= (\kappa(\varepsilon, \varepsilon),
    \gamma^F(\varepsilon,\varepsilon)).
  \end{split}
\end{equation}
As in every Lie superalgebra $\fk=\fk_{\bar 0}\oplus\fk_{\bar 1}$, the
odd subspace $\fk_{\bar 1}$ generates an ideal
$\widehat\fk=\widehat\fk_{\bar 0}\oplus\widehat\fk_{\bar 1}$ of $\fk$,
where $\widehat\fk_{\bar 0}=[\fk_{\bar 1},\fk_{\bar 1}]$ and
$\widehat\fk_{\bar 1}=\fk_{\bar 1}$. We refer to it as the
\emph{Killing ideal} of the Killing superalgebra.

\subsection{The Killing superalgebra is a filtered subdeformation of $\fp$}
\label{sec:admfilsubdef}

We now show that the Lie superalgebra $\fk$ described in
\eqref{eq:KSA} is isomorphic to a filtered subdeformation of the
Poincaré superalgebra $\fp$. It is not in general a subalgebra of
$\fp$.

We will first show that the Killing superalgebra $\fk$ determines a
$\mathbb Z$-graded vector subspace of $\fp$. As above, let us identify
$\fp$ as a vector space with the fibre $\eE_o$. Let
$\ev_o^{\bar 0}: \fk_{\bar 0} \to V$ be the composition of evaluation
at $o$ and projection onto $V=T_o M$, and let
$\ev_o^{\bar 1}: \fk_{\bar 1} \to S$ be evaluation at $o$. Let also
$\im\ev_o^{\bar 1} = S' \subseteq S$ and $\im \ev_o^{\bar 0} = V'
\subseteq V$ be the images of these evaluations.

Let $\fh = \ker \ev_o^{\bar 0}$ be the Lie subalgebra of
$\fk_{\bar 0}$ consisting of elements of $\fk_{\bar 0}$ which vanish
at $o \in M$; that is, which take the form
$(0,A) \in V \oplus \fso(V)$.  In other words, $\fh$ defines a
subspace of $\fso(V)$.  From the definition of $\fh$, we have a short
exact sequence of vector spaces
\begin{equation*}
  \begin{CD}
    0 @>>> \fh @>>> \fk_{\bar 0} @>{\ev_o^{\bar 0}}>> V'
    @>>> 0.
  \end{CD}
\end{equation*}
Since short exact sequences split in the category of vector spaces, we
have a vector space isomorphism $\fk_{\bar 0} \cong \fh \oplus V'$ and
since $\fk_{\bar 1} \cong S'$, we see that (again, as a vector space)
$\fk \cong \fh \oplus S' \oplus V' \subset \fso(V) \oplus S \oplus V =
\fp$.  However there is no canonical splitting and hence no preferred
isomorphism.  If there were, we could simply transport the Lie
superalgebra structure in $\fk$ to $\fh \oplus S' \oplus V'$.  In our
case, however, we will have to choose a splitting.  Geometrically,
this amounts to choosing (in a linear fashion) for every $v \in V'$ a
Killing vector field $\xi \in \fk_{\bar 0}$ with
$\ev_o^{\bar 0}(\xi)= v$.  Such a choice gives an embedding of $V'$
into $\fk_{\bar 0}\subset V\oplus \fso(V)$ by sending $v \in V'$ to
$(v, X_v)$, where $X_v \in \fso(V)$ is the image of $v$ under a linear
map $X: V' \to \fso(V)$.  Any other choice of splitting would result
in $(v,X'_v)$ for some other linear map $X': V' \to \fso(V)$, but
where $X - X' : V' \to \fh$. (In applications we can fix
$X:V'\to \fh^\perp$ to take values in a complement $\fh^\perp$ of
$\fh$ in $\fso(V)$, e.g., the orthogonal complement with respect to
the Killing form, whenever this exists).

Since $\fh$ consists of those Killing vectors in $\fk_{\bar 0}$ which
vanish at $o$, the corresponding parallel section of $\eE_{\bar 0}$
takes the form $(0,A)$ at $o\in M$.  Now it is clear from equation
\eqref{eq:KSA} that the Lie brackets of $\fk$ only depend on the value
of the sections at the point $o \in M$, hence we see that if
$(0,A),(0,B) \in \fh$, then
\begin{equation*}
  [(0,A),(0,B)] = (0, [A,B]),
\end{equation*}
so that $\fh$ defines a Lie subalgebra of $\fso(V)$. In addition, when
evaluated at $o \in M$, the condition $\eL_\xi F = 0$ for all vector
fields $\xi$ with $\ev_o^{\bar 0}(\xi)=0$ becomes
\begin{equation*}
  A\cdot F_o = 0\quad\text{for all}~A\in\fh.
\end{equation*}
In summary, $\fh$ defines a Lie subalgebra of $\fso(V) \cap
\stab(\varphi)$, where $\varphi = F_o \in \Lambda^4 V^*\cong\Lambda^4
V$ is the value of $F$ at $o \in M$.

It also follows from equation~\eqref{eq:KSA} that the action of $\fh$
on $\fk_{\bar 1}$ at $o \in M$ is the restriction to $\fh < \fso(V)$
of the action of $\fso(V)$ on $S$:
\begin{equation*}
  [(0,A),s] = \sigma(A)s.
\end{equation*}
This implies in particular that $\fh$ preserves the subspace $S'$.

Similarly, using the fixed embedding $V' \subset V \oplus \fso(V)$
given by $v \mapsto (v, X_v)$, we find that the remaining brackets are
\begin{equation*}
  \begin{split}
    [(0,A), (v, X_v)] &= (Av, [A,X_v]) = (Av, X_{Av}) + (0, [A, X_v] - X_{Av})\\
    [(v,X_v), s] &= \beta^\varphi(v,s) + X_v s\\
    [s,s] &= (\kappa(s,s), X_{\kappa(s,s)}) + (0,
    \gamma^{\varphi}(s,s) - X_{\kappa(s,s)})
  \end{split}
\end{equation*}
and
\begin{equation*}
  \begin{split}
    [(v,X_v), (w,X_w)] &= ( X_vw - X_w v, [X_v,X_w] + R(v,w))\\
    &= (X_vw - X_wv, X_{X_vw - X_wv}) + (0, [X_v,X_w] - X_{X_vw - X_wv} + R(v,w)).
  \end{split} 
\end{equation*}

In summary, the Killing superalgebra $\fk$ is isomorphic to a Lie
superalgebra structure defined on the graded subspace
$\fh \oplus S' \oplus V'$ of $\fp$, where $\fh$ is a Lie subalgebra of
$\fso(V) \cap \stab(\varphi)$ which preserves $S'$. The Lie brackets
are the following:
\begin{equation}
  \label{eq:KSAina}
  \begin{aligned}[m]
	 [A,B] &= AB - BA\\
      [A,s] &= \sigma(A)s\\
    [A,v] &= Av + \delta(A,v)\\
  \end{aligned}
  \qquad\qquad
  \begin{aligned}[m]
    [s,s] &= \kappa(s,s) + \gamma^\varphi(s,s) - X_{\kappa(s,s)}\\
      [v,s] &= \beta^\varphi(v,s) + \sigma(X_v)s\\
   [v,w] &= \alpha(v,w) + \rho(v,w),
  \end{aligned}
\end{equation}
where $A,B \in \fh$, $v,w \in V'$, $s \in
S'$ and 
\begin{equation*}
  \begin{split}
	\alpha(v,w) &= X_v w - X_w v\\
    \delta(A,v) &= [A,X_v] - X_{Av}\\    
		\rho(v,w) &= [X_v,X_w]- X_{\alpha(v,w)} + R(v,w).
  \end{split}
\end{equation*}
We observe that both $\delta: \fh \otimes V' \to \fh$ and
$\rho:\Lambda^2V' \to \fh$ take values in $\fh$.  Moreover the element
$\beta(v,s)=\beta^\varphi(v,s)+\sigma(X_{v})s$ is in $S'$ (and not
$S$) whilst the individual terms may not; similarly the sum
$\gamma(s,s)=\gamma^\varphi(s,s)-X_{\kappa(s,s)}$ is in $\fh$ (and not
$\fso(V)$).

From now on we will identify $\fk$ with the Lie superalgebra structure
defined on the graded subspace $\fh \oplus S' \oplus V'$ of $\fp$ by
\eqref{eq:KSAina}. The grading of the Poincaré superalgebra $\fp$
gives rise to a natural filtration of $\fp$:
\begin{equation*}
  \fp^\bullet : \qquad \fp = \fp^{-2} \supset \fp^{-1} \supset \fp^0
  \supset 0,
\end{equation*}
where $\fp^{-1} = \fso(V) \oplus S$, $\fp^0 = \fso(V)$, and therefore
also to a filtration of $\fk$
\begin{equation*}
  \fk^\bullet : \qquad \fk = \fk^{-2} \supset \fk^{-1} \supset \fk^0
  \supset 0,
\end{equation*}
where $\fk^{-1} = \fh \oplus S'$, $\fk^0 = \fh$.  One checks from the
Lie brackets \eqref{eq:KSAina} that
$[\fk^i, \fk^j] \subset \fk^{i+j}$, so that $\fk^\bullet$ is a
filtered Lie superalgebra. Its associated graded Lie superalgebra
$\fk_\bullet$ has graded pieces $\fk_{-2} = V'$, $\fk_{-1} = S'$,
$\fk_0 = \fh$ and, comparing again with the Lie brackets of $\fk$, we
see that $\fk_\bullet$ is a subalgebra of $\fp$.  Indeed, the maps
$\alpha, \beta, \gamma, \delta, \rho$ all have positive filtration
degree (compare also with equation \eqref{eq:generalbrackets}).

Of course, it is not an arbitrary filtered subdeformation, since some
of its components are prescribed by the supergravity theory via the
definition of Killing spinor (compare also with equation
\eqref{eq:generalbracketsII}).  

In summary, we have proved most of the following

\begin{theorem}
  \label{thm:KSAisFSD}
  The Killing superalgebra $\fk$ is a filtered subdeformation of the
  Poincaré superalgebra and if $\dim \fk_{\bar 1} > 16$ it is a
  realizable filtered subdeformation.  Moreover the Killing ideal
  $\widehat\fk$ is odd-generated realizable.
\end{theorem}

\begin{proof}
  It remains to prove that $\fk$ is realizable if
  $\dim\fk_{\bar 1}>16$, from where it will follow that $\widehat\fk$
  is odd-generated and realizable.  Property~(i) of
  Definition~\ref{def:realizable} is immediate, whereas property~(ii)
  follows from the fact that the exterior derivative $dF$ of a
  $\fk_{\bar 0}$-invariant $4$-form $F$ is a $\fk_{\bar 0}$-invariant
  $5$-form, hence (locally) determined by its value $(dF)_o$ at
  $o\in M$.
\end{proof}

\subsection{Highly supersymmetric lorentzian spin manifolds}
\label{sec:admfilsubdefII}

We will now restrict to the highly supersymmetric case and show that
any realizable filtered subdeformation
$\fg=\fg_{\bar 0}\oplus\fg_{\bar 1}$ of $\fp$ can be realised as (a
subalgebra of) the Killing superalgebra of a homogeneous
$(M=G_{\bar 0}/H,g,F)$.  To this end, it is actually more natural to
assume $\fg$ to be \emph{anti}-isomorphic to a realizable filtered
subdeformation; in other words, in this section $\fg$ has the opposite
Lie brackets to those in
equations~\eqref{eq:generalbrackets}~and~\eqref{eq:generalbracketsII}.

We first need to recall some basic definitions. Let $G$ be a connected
Lie supergroup with Lie superalgebra $Lie(G)=\fg$. We consider it as a
super Harish-Chandra pair \cite{MR0580292, MR760837, MR2640006}, a
pair $G=(G_{\bar 0},\fg)$ consisting of a connected Lie group
$G_{\bar 0}$ with Lie algebra $Lie(G_{\bar 0})=\fg_{\bar 0}$ and a Lie
superalgebra $\fg=\fg_{\bar 0}\oplus\fg_{\bar 1}$ admitting an
\emph{adjoint representation}, i.e., a morphism of Lie groups
\begin{equation}
\label{eq:adjoint}
\Ad:G_{\bar 0}\longrightarrow\GL(\fg)
\end{equation}
such that $\ad(x)y=\frac{d}{dt}|_{t=0}\Ad_{\exp(tx)}y$ for all
$x\in\fg_{\bar 0}$ and $y\in\fg$.  In particular $V'=V$ and the
analytic subgroup $H$ of $G_{\bar 0}$ with Lie algebra $Lie(H)=\fh$
acts orthogonally on $V\cong \fg_{\bar 0}/\fh$ via the natural
representation
\begin{equation}
  \label{eq:adjoint2}
  \Ad:H\longrightarrow\SO(V) 
\end{equation}
induced by \eqref{eq:adjoint}. 

If $H$ is closed in $G_{\bar 0}$, then $M=(G_{\bar 0}/H,g)$ is an
$11$-dimensional lorentzian homogeneous manifold, where $g$ is the
$G_{\bar 0}$-invariant lorentzian metric on $M$ corresponding to the
$H$-invariant inner product $\eta$ on $V$. Consider the
$\SO(V)$-bundle $P$ on $M$ of oriented orthonormal frames of
$(M,g)$. We have
\begin{equation*}
  P\cong G_{\bar 0}\times_{H} \SO(V) 
\end{equation*}
and $TM\cong P\times_{SO(V)} V\cong G_{\bar 0}\times_{H} V$. In
particular the vector fields on $M$ are identified with the
$H$-equivariant maps $\xi:G_{\bar 0}\to V$.

Any lift of the adjoint representation \eqref{eq:adjoint2} to the spin
group $\Spin(V)$ --- i.e., any homomorphism $H\to\Spin(V)$ such that the
diagram
\begin{equation}
\label{eq:liftspin}
\begin{CD}
H @> >> \Spin(V) \\
@|             @VV \sigma V    \\
H     @>\Ad >>    \SO(V) 
\end{CD}
\end{equation}
commutes --- allows us to define a spin structure
$Q=G_{\bar 0}\times_{H} \Spin(V)$ on $(M,g)$, usually referred to as the
homogeneous spin structure associated to the lift \eqref{eq:liftspin}
\cite{MR1164864}. If $G_{\bar 0}$ is simply connected, the homogeneous
spin structures are in one-to-one correspondence with the spin
structures \cite{MR1166019}. Now, since $\fa$ is transitive and
fundamental, any element $A\in\fh$ is uniquely determined by its
action on $\fg_{\bar 1}\simeq S'\subseteq S$ and it is not difficult
to see that the restriction of \eqref{eq:adjoint} to $H$ and
$\fg_{\bar 1}$ determines a unique lift
$\Ad:H\longrightarrow \Spin(V)$.  We call any triple
$(M=G_{\bar 0}/H,g,Q)$ with $Q$ determined by \eqref{eq:adjoint} as
above a \emph{homogeneous lorentzian spin manifold associated with}
$\fg$. For an analogous discussion in the special case of reductive
homogeneous manifolds with $S'=S$ we refer the reader to
\cite[\S\S5.1-2]{MR2640006} and \cite{ALOKilling}.  The spin bundle
on $M$ is $\$=Q\times_{\Spin(V)}S\cong G_{\bar 0}\times_{H} S$ and the
spinor fields on $M$ are identified with the $H$-equivariant maps
$\varepsilon:G_{\bar 0}\to S$.

Finally, it is often convenient to work on $G_{\bar 0}$ through the
natural projection $\pi:G_{\bar 0}\to M=G_{\bar 0}/H$.  For instance
invariant affine connections on $M=G_{\bar 0}/H$ are known to be in a
one-to-one correspondence with Nomizu maps; that is, linear maps
\begin{equation*}
  L:\fg_{\bar 0}\to\mathfrak{gl}(V)
\end{equation*}
satisfying \cite{MR1393941}:
\begin{enumerate}[label=(\roman*)]
\item $L(A)=\ad(A)$ for all $A\in\fh$; and
\item $L$ is $H$-equivariant.
\end{enumerate}
Let us consider the natural projection from $\fg_{\bar 0}$ to
$V\cong\fg_{\bar 0}/\fh$ and trivially extend $\eta$ to the
$H$-invariant symmetric bilinear map
$(-,-):\fg_{\bar 0}\otimes \fg_{\bar 0}\to\mathbb R$ with kernel
$\fh$, and let $U$ be the symmetric bilinear map on $\fg_{\bar 0}$
with values into $V$ uniquely determined by
\begin{equation*}
  2(U(x,y),z)=(x,[z,y])+([z,x],y),
\end{equation*}
where $x,y,z\in\fg_{\bar 0}$. It is not difficult to see that the
operator $\widetilde L:\fg_{\bar 0}\to\Hom(\fg_{\bar 0},V)$ given by
\begin{equation*}
  \widetilde L(x)y:=\tfrac12 [x,y]\!\!\!\!\!\mod\fh + U(x,y)
\end{equation*}
factors through a Nomizu map $L:\fg_{\bar 0}\to\mathfrak{gl}(V)$ which
satisfies
\begin{enumerate}[label=(\roman*),start=3]
\item $\im L\subseteq\fso(V)$;
\item $\widetilde L(x)y-\widetilde L(y)x-[x,y]\!\!\!\mod\fh=0$ for all
  $x,y\in\fg_{\bar 0}$.
\end{enumerate}
Indeed this is the Nomizu map associated to the Levi-Civita connection
of (M,g) (cf. \cite[Theorem 3.3]{MR1393941} for the case of reductive
homogeneous manifolds).

The Levi-Civita covariant derivative can be easily described, at least
locally. Let $\xi_i:G_{\bar 0}\to V\cong\fg_{\bar 0}/\fh$ be (locally
defined) vector fields on $M$, $i=1,2$, and choose (locally defined)
vector fields $\widetilde\xi_i:G_{\bar 0}\to \fg_{\bar 0}$ on
$G_{\bar 0}$ such that $\xi_i$ is $\pi$-related to $\widetilde\xi_i$,
i.e., such that $\pi_*(\widetilde\xi_i)=\xi_i$ for $i=1,2$.  Then
\begin{equation}
  \label{eq:exform}
  \nabla_{\xi_1}\xi_2=\pi_*(\widetilde\xi_1(\widetilde\xi_2) +
  \widetilde L (\widetilde\xi_1)(\widetilde\xi_2)),
\end{equation}
where $\widetilde\xi_1(\widetilde\xi_2)$ is the derivative of
$\widetilde\xi_2$ along $\widetilde\xi_1$ and $\widetilde L$ acts as
usual at any fixed $g\in G_{\bar 0}$. For more details, we refer the
reader to e.g. \cite[\S4]{MR2640006}.

We are now ready to state our main result, which subsumes
Theorem~\ref{thm:KSAisFSD}.

\begin{theorem}
  \label{thm:firstmain}
  Let $(M,g,F)$ be an $11$-dimensional lorentzian spin manifold
  endowed with a closed $F\in\Omega^4(M)$ and $\fk=\fk_{\bar
    0}\oplus\fk_{\bar 1}$ the associated Killing superalgebra.  If
  $\dim\fk_{\bar 1}>16$ then $(M,g,F)$ is locally homogeneous and the
  Killing superalgebra $\fk$ (resp. Killing ideal $\widehat\fk$) is a
  (resp. odd-generated) realizable filtered subdeformation of
  $\fp$.

  Conversely, let $\fg=\fg_{\bar 0}\oplus\fg_{\bar 1}$ be (the
  opposite Lie superalgebra to) a realizable filtered subdeformation
  of $\fp$, with corresponding $11$-dimensional homogeneous lorentzian
  spin manifold $(M=G_{\bar 0}/H,g,Q)$. Then there exist
  \begin{enumerate}[label=(\arabic*)]
  \item a $G_{\bar 0}$-invariant closed $4$-form $F\in\Omega^4(M)$;
  \item an (anti)embedding $\Phi:\fg\to\fk$ of realizable filtered
    subdeformations of $\fp$ from $\fg$ in the Killing superalgebra
    $\fk$ of $(M,g,F)$. If $\fg$ is odd-generated realizable, then
    $\Phi(\fg)\subseteq\widehat\fk$.
  \end{enumerate}
  In particular $\dim\fk_{\bar 1}>16$. 
\end{theorem}

\begin{proof}
  The first statement is a direct consequence of the local homogeneity
  theorem in \cite{FigueroaO'Farrill:2012fp} and
  Theorem~\ref{thm:KSAisFSD} in Section~\ref{sec:admfilsubdef}.

  Let now $\fg=\fg_{\bar 0}\oplus\fg_{\bar 1}$ be the opposite Lie
  superalgebra to a realizable filtered deformation of a graded
  subalgebra $\fa=\fh\oplus S'\oplus V$ of $\fp$ and
  $(M=G_{\bar 0}/H,g,Q)$ an associated homogeneous lorentzian spin
  manifold.  Since the Lie brackets of $\fg$ are the opposite of those
  in \eqref{eq:generalbrackets} and \eqref{eq:generalbracketsII}, we
  have that the map $L:\fg_{\bar 0}\to\fso(V)$,
  \begin{equation}
    \label{eq:LC}
    \begin{split}
      L(A)&=\ad(A)\qquad(A\in\fh)\\
      L(v)&=-X_v\qquad(v\in V),
    \end{split}
  \end{equation}
  satisfies properties (i)-(iv) and therefore is the Nomizu map
  corresponding to the Levi-Civita connection of $(M,g)$.

  Consider the fundamental vector field    
  \begin{equation}
    \label{eq:Killeven}
    \xi^{(x)}:G_{\bar 0}\to V,\qquad \xi^{(x)}(g)=(\Ad_{g^{-1}}x)\!\!\!\mod\fh
  \end{equation}
  associated to $x=(A,v)\in \fg_{\bar 0}\cong \fh\oplus V$. Clearly
  $\xi^{(x)}$ is a Killing vector field and using equations
  \eqref{eq:exform}, \eqref{eq:LC} and \eqref{eq:Killeven}, it can be
  checked directly that the value of the section
  $(\xi^{(x)},-\nabla\xi^{(x)})$ of $\eE_{\bar 0} = TM \oplus
  \fso(TM)$ at $o=eH\in M$ is $(v,A+X_v)\in V\oplus\fso(V)$.  This
  gives the realisation of the abstract Lie algebra $\fg_{\bar 0}$ as
  subalgebra of the algebra of Killing vector fields on $M$:
  \begin{equation}
    \label{eq:anti1}
    [\xi^{(x)},\xi^{(y)}]=-\xi^{([x,y])},
  \end{equation}
  for all $x,y\in\fg_{\bar 0}$.  Given any admissible
  $\varphi\in\Lambda^4 V^*\cong \Lambda^4 T^*_oM$, we let
  $F\in\Omega^4(M)$ be the unique $G_{\bar 0}$-invariant closed
  $4$-form with value $F_o=\varphi$ at $o\in M$.  As for the elements
  of $\fg_{\bar 0}$, every $s\in \fg_{\bar 1}\cong S'\subseteq S$ has
  an associated spinor field
  \begin{equation}
    \label{eq:Killodd}
    \varepsilon^{(s)}:G_{\bar 0}\to S,\qquad\varepsilon^{(s)}(g)=\Ad_{g^{-1}}s.
  \end{equation}
  For any vector field $\xi:G_{\bar 0}\to V$ with $\pi$-related
  $\widetilde\xi:G_{\bar 0}\to \fg_{\bar 0}$ we compute
  \begin{equation*}
    \begin{split}
      \nabla_{\xi}\varepsilon^{(s)} &=
      \pi_*(\widetilde\xi(\varepsilon^{(s)}) + \sigma(\widetilde
      L(\widetilde\xi))(\varepsilon^{(s)}))\\
      &=\pi_*(-\ad(\widetilde\xi)(\varepsilon^{(s)})+\sigma(\widetilde
      L(\widetilde\xi))(\varepsilon^{(s)}))\\
      &=\beta^{\varphi}(\xi,\varepsilon^{(s)})
    \end{split}
  \end{equation*}
  where the last equality follows from the Lie brackets of
  $\fg_{\bar 0}$ and \eqref{eq:LC}. This shows that
  $\varepsilon^{(s)}$ is a Killing spinor, for all $s\in S'$.

  The required map $\Phi:\fg\to\fk$ is defined by:
  \begin{equation*}
      \Phi(x)=\xi^{(x)} \qquad\text{and}\qquad \Phi(s)=\varepsilon^{(s)},
  \end{equation*}
  where $x=(v,A)\in \fg_{\bar 0}$ and $s\in\fg_{\bar 1}$. Note that
  \begin{equation*}
    \begin{split}
      \eL_{\xi^{(x)}} \varepsilon^{(s)} &= \nabla_{\xi^{(x)}}
      \varepsilon^{(s)} - \sigma(\nabla\xi^{(x)}) \varepsilon^{(s)}\\
      &= \beta^{\varphi}(\xi^{(x)},\varepsilon^{(s)}) -
      \sigma(\nabla\xi^{(x)})\varepsilon^{(s)}
    \end{split}
  \end{equation*}
  so that $\eL_{\xi^{(x)}}\varepsilon^{(s)}$ is the Killing spinor on
  $M$ with value $\beta^{\varphi}(v,s)+\sigma(A)s+\sigma(X_v)s$ at
  $o\in M$. In other words
  \begin{equation}
    \label{eq:anti2}
    [\xi^{(x)},\varepsilon^{(s)}]=-\varepsilon^{([x,s])}
  \end{equation}
  and one similarly checks
  \begin{equation}
    \label{eq:anti3}
    [\varepsilon^{(s)},\varepsilon^{(s)}]=-\xi^{([s,s])}.
  \end{equation}
  Identities \eqref{eq:anti1}, \eqref{eq:anti2} and \eqref{eq:anti3}
  show that $\Phi$ is a Lie superalgebra anti-homomorphism. The fact
  that $\Phi$ is an (anti)embedding of realizable filtered
  subdeformations of $\fp$ is immediate, as well as the last two
  claims of the theorem.
\end{proof}

\begin{remark}
  The $G_{\bar 0}$-invariant closed $F\in\Omega^4(M)$ associated to a
  realizable filtered deformation in Theorem \ref{thm:firstmain} is a
  priori not unique, since it appears to depend on the choice of an
  admissible $\varphi\in\Lambda^4 V$ (recall Definition
  \ref{def:realizable}). However, as already advertised, we will
  obtain $\ker \imath^*=0$ in Corollary \ref{cor:final!}, so that
  $\varphi$ (and $F$) are unique.
\end{remark}

\begin{remark}
  The Killing superalgebra $\fk$ of the homogeneous lorentzian spin
  manifold $(M,g,Q,F)$ associated to $\fg$ in Theorem
  \ref{thm:firstmain} is strictly larger than $\fg$ in general. (The
  analogous statement holds for Killing ideals $\widehat\fk$ and
  odd-generated realizable $\fg$.) We do not know of general conditions on
  $\fg$ under which equality actually holds.
\end{remark}

Theorem \ref{thm:firstmain} and the above remarks say that Killing
superalgebras (resp. Killing ideals) of highly supersymmetric
$(M,g,F)$, up to local equivalence, are in a one-to-one correspondence
with \emph{maximal} realizable (resp. odd-generated realizable) filtered
subdeformations of $\fp$, up to isomorphism of filtered
subdeformations.

In Sections~\ref{sec:class-probl-kill} and \ref{sec:field-eqns} below,
we set up the classification problem for the Killing superalgebras of
highly supersymmetric $11$-dimensional supergravity backgrounds as the
classification problem of realizable filtered subdeformations of
$\fp$. In particular, we show that high supersymmetry implies that
the Einstein and Maxwell equations are satisfied; that is, the
homogeneous lorentzian spin manifold reconstructed in
Theorem \ref{thm:firstmain} from a realizable filtered subdeformation
is automatically a supergravity background.

In this regard, we remark that one needs the full datum of a
realizable filtered subdeformation of $\fp$ to reconstruct the
supergravity background unambiguously; the assignment of a Lie
superalgebra is not sufficient in general. For instance there is an
example of a Lie superalgebra with (at least) two \emph{non-isomorphic}
structures of realizable filtered subdeformation of $\fp$: namely, the
Killing superalgebra of a supergravity background with $24$
supercharges described in \cite{Michelson:2002wa} and shown in
\cite{Fernando:2004jt} to be isomorphic \emph{as an abstract Lie
superalgebra} to a subalgebra of the Killing superalgebra of the
maximally supersymmetric pp-wave of \cite{KG}.

\section{The classification problem for Killing superalgebras}
\label{sec:class-probl-kill}

We have just seen that the Killing superalgebra is a filtered
subdeformation of the Poincaré superalgebra. In the highly
supersymmetric case, Proposition \ref{prop:fil1} applies and the aim
of this section is to improve that result in the case of (odd-generated)
realizable filtered subdeformations.

\subsection{The Jacobi identity of Killing superalgebras}
\label{sec:JacobiIdentities}

The Lie brackets of a Killing superalgebra are
given by equation~\eqref{eq:KSAina} in terms of the following data.

First we have a graded Lie subalgebra $\fa = \fh \oplus S' \oplus V'$
of the Poincaré superalgebra. In particular, this means that
$\fh < \fso(V)$ stabilises both $S' \subseteq S$ and $V' \subseteq V$
and that $\kappa(S',S') \subseteq V'$.  The rest of the data consists
of an $\fh$-invariant $\varphi\in \Lambda^4 V$, $X: V' \to \fso(V)$
(or, more precisely, $V' \to \fso(V)/\fh$) and
$R : \Lambda^2 V' \to \fso(V)$.  In terms of this data, we have the
following Lie brackets on the vector space $\fh \oplus S' \oplus V'$:
\begin{equation}
  \label{eq:preKSA}
  \begin{split}
    [A,B] &= AB - BA\\
    [A,s] &= \sigma(A)s\\
    [A,v] &= Av + [A, X_v] - X_{Av}\\
    [s,s] &= \kappa(s,s) + \gamma^{\varphi}(s,s) - X_{\kappa(s,s)}\\
    [v,s] &= \beta^{\varphi}(v,s) + \sigma(X_v)s \\
    [v,w] &= X_v w - X_w v + [X_v,X_w] - X_{X_v w - X_w v} + R(v,w),
  \end{split}
\end{equation}
where $A,B \in \fh$, $v,w \in V'$ and $s \in S'$. It bears reminding
that the right-hand sides of the Lie brackets in \eqref{eq:preKSA}
take values in $\fh \oplus S' \oplus V'$, but that the individual
terms may not.  For example, $[A,X_v], X_{Av} \in \fso(V)$, but their
difference $[A,X_v] - X_{Av} \in \fh$.  Similarly,
$\gamma^{\varphi}(s,s) - X_{\kappa(s,s)}\in \fh$, but
$\gamma^{\varphi}(s,s), X_{\kappa(s,s)}\in \fso(V)$, and the same
happens with $[X_v,X_w] - X_{X_v w - X_w v} + R(v,w) \in \fh$, even
though $[X_v,X_w], X_{X_v w - X_w v}, R(v,w) \in \fso(V)$. Also
$X_v w - X_w v \in V'$, $\beta^{\varphi}(v,s)+\sigma(X_v)s \in S'$,
but $X_v w \in V$ and $\beta^{\varphi}(v,s), \sigma(X_v)s \in S$.

The only additional conditions come from demanding that the Lie
brackets \eqref{eq:preKSA} do define a Lie superalgebra.  In other
words, they come from imposing the Jacobi identity.  There are ten
components of the Jacobi identity and we summarise the results for
each component in turn.

\subsubsection*{The $[\fh\fh\fh]$ Jacobi}

This is automatically satisfied because $\fh$ is a Lie subalgebra of
$\fso(V)$.

\subsubsection*{The $[\fh\fh S']$ Jacobi}

This is automatically satisfied because the action of $\fh$ on $S'$ is
the restriction to $\fh$ and $S'$ of the spinor representation
$\sigma$ of $\fso(V)$ on $S$.

\subsubsection*{The $[\fh\fh V']$ Jacobi}

This is also automatically satisfied, extending the adjoint action of
$\fh$ on itself to $\fso(V)$.

\subsubsection*{The $[\fh S' S']$ Jacobi}

This is automatically satisfied since
$\fh < \fso(V)\cap\stab(\varphi)$.  Indeed, for $A \in \fh$ and
$s \in S'$,
\begin{equation*}
  \begin{split}
    [A, [s,s]] &= [A, \kappa(s,s) + \gamma^\varphi(s,s) - X_{\kappa(s,s)}]\\
    &= A \kappa(s,s) + [A, \gamma^\varphi(s,s)] - X_{A\kappa(s,s)},
  \end{split}
\end{equation*}
whereas
\begin{equation*}
  2 [[A,s], s] = 2 [\sigma(A)s,s] = 2 \kappa(\sigma(A)s,s) + 2
  \gamma^\varphi(\sigma(A)s,s) - 2 X_{\kappa(\sigma(A)s,s)}.
\end{equation*}
Since $\fh < \fso(V)$, $A \kappa(s,s) = 2 \kappa(\sigma(A)s,s)$, so that the
Jacobi identity is satisfied provided that
\begin{equation*}
  [A, \gamma^\varphi(s,s)] =  2 \gamma^\varphi(\sigma(A)s,s).
\end{equation*}
But $\gamma^\varphi$ only depends on $\fso(V)$-equivariant operations
like Clifford product and Dirac current, and on $\varphi$. It follows
that $\gamma^\varphi$ is equivariant under
$\fso(V) \cap \stab(\varphi)$, and by restriction also under $\fh$.

\subsubsection*{The $[\fh S' V']$ Jacobi}

In this case, for $A \in \fh$, $v \in V'$ and $s \in S'$,
\begin{equation*}
  [A,[v,s]] - [[A,v],s] - [v,[A,s]] = \sigma(A) \beta^\varphi(v,s) -
  \beta^\varphi(Av,s) - \beta^\varphi(v,\sigma(A)s).
\end{equation*}
The Jacobi identity is again satisfied since $\fh < \fso(V) \cap
\stab(\varphi)$.

\subsubsection*{The $[\fh V' V']$ Jacobi}

A somewhat lengthy calculation shows that, for all $A \in \fh$ and
$v,w \in V'$,
\begin{equation*}
  [A, [v,w]] - [[A,v],w] - [v,[A,w]] = [A,R(v,w)] - R(Av,w) - R(v,Aw).
\end{equation*}
It follows that the Jacobi identity is satisfied if and only if
\begin{equation}
  \label{eq:Requiv}
 R : \Lambda^2 V' \to \fso(V) \quad\text{is $\fh$-equivariant}.
\end{equation}

\subsubsection*{The $[S'S'S']$ Jacobi}

The Jacobi identity says that $[[s,s],s] = 0$ for all $s \in S'$, and
it expands to
\begin{equation*}
 \sigma(\gamma^\varphi(s,s)) s = - \beta^\varphi(\kappa(s,s),s),
\end{equation*}
for all $s\in S'$. This identity is known to be automatically
satisfied for all $s\in S$,
cf. \cite[Proposition~7]{Figueroa-O'Farrill:2015efc}.

\subsubsection*{The $[S'S'V']$ Jacobi}

After another somewhat lengthy calculation, and letting 
\begin{equation*}
  \beta^\varphi_v(s) = \beta^\varphi(v,s)
\end{equation*}
for all $v\in V$ and $s\in S$, the Jacobi identity is
equivalent to
\begin{equation}
  \label{eq:ssvJac}
  \begin{split}
    \tfrac12 R(v,\kappa(s,s))w &=
    \kappa((X_v\beta^\varphi)(w,s),s) +
    \gamma^\varphi(\beta^{\varphi}_{v}s,s)(w)\\
    &=\kappa((X_v\beta^\varphi)(w,s),s) -
    \kappa(\beta^\varphi_v(s) , \beta^\varphi_w(s)) -
    \kappa(\beta^\varphi_w \beta^\varphi_v(s),s),
  \end{split}
\end{equation}
for all $s \in S'$, $v \in V'$ and $w \in V$. 

\begin{remark}
  If $\dim S'>16$, then by local homogeneity $V'=V$, and equation
  \eqref{eq:ssvJac} expresses the curvature operator
  $R: \Lambda^2 V\to \fso(V)$ in terms of $X$ and $\varphi$. By a
  further contraction, this determines the Ricci tensor and as we will
  show in Section~\ref{sec:field-eqns}, it implies the bosonic field
  equations of $11$-dimensional supergravity.
\end{remark}

\subsubsection*{The $[S'V'V']$ Jacobi}

This Jacobi identity expands to the following condition
\begin{equation}
  \label{eq:svvJac}
  R(v,w)s = (X_v\beta^\varphi)(w,s) - (X_w\beta^\varphi)(v,s) +
  [\beta^\varphi_v,\beta^\varphi_w](s),
\end{equation}
for all $s \in S'$ and $v,w \in V'$.

\subsubsection*{The $[V'V'V']$ Jacobi}

Finally the last component of the Jacobi identity expands to two
Bianchi-like identities, one algebraic
\begin{equation}
  \label{eq:algBianchi}
  R(u,v)w + R(v,w) u + R(w,u) v = 0,
\end{equation}
and one
differential
\begin{multline}
  \label{eq:difBianchi}
  R(X_u v - X_v u, w) + R(X_v w - X_w v, u) + R(X_w u - X_u w, v)\\
  = [X_w,R(u,v)] + [X_u,R(v,w)] + [X_v,R(w,u)],
\end{multline}
for all $u,v,w \in V'$. If $V' = V$, \eqref{eq:algBianchi} is
precisely the algebraic Bianchi identity for $R$, whereas the
differential identity simplifies to
\begin{equation}
  \label{eq:difBianchiVp=V}
  (X_u R)(v,w) + (X_v R)(w,u) + (X_w R)(u,v) = 0.
\end{equation}
(Notice that $X_u \in \fso(V)$, but unless $V' = V$, $R \in
\Hom(\Lambda^2 V', \fso(V))$, which is not an $\fso(V)$-module, but
only an $\fh$-module.)

\subsection{The classification problem for highly supersymmetric
  Killing superalgebras}
\label{sec:class-probl-highly}

Particularly interesting is the highly supersymmetric case, where
$\dim S'> 16$ so that $V'=V$.  In this case, the classification
problem for highly supersymmetric Killing superalgebras breaks up into
two main steps:
\begin{enumerate}
\item classify highly supersymmetric graded subalgebras $\fa = \fh
  \oplus S' \oplus V$ of the Poincaré superalgebra $\fp$;
\item for each such $\fa$, find $\varphi \in (\Lambda^4 V)^\fh$, $R \in
  \Hom(\Lambda^2 V, \fso(V))^\fh$ which is an algebraic curvature tensor (i.e.,
  satisfying the algebraic Bianchi identity~\eqref{eq:algBianchi}) and
  $X: V\to \fso(V)$ (only its image modulo $\fh$ matters) such that:
  \begin{enumerate}[label=(\roman*)]
  \item $\varphi$ is closed, cf. (ii) of
    Definition~\ref{def:realizable};
  \item the right-hand sides of the expressions in \eqref{eq:preKSA}
    take values in $\fh\oplus S'\oplus V$;
  \item the three equations~\eqref{eq:ssvJac}, \eqref{eq:svvJac} and
    \eqref{eq:difBianchiVp=V} are satisfied.
  \end{enumerate}
\end{enumerate}
The Jacobi identity \eqref{eq:ssvJac} determines $R$ in terms of
$\varphi$ and $X$ so that the highly supersymmetric
Killing superalgebra (or, more generally, any realizable filtered
subdeformation of $\fp$) is completely determined by
$(\fh, S', \varphi, X)$.  This result improves Proposition
\ref{prop:fil1} in the case of \emph{realizable} filtered
subdeformations.

Step (1) of the classification problem is too broad and not tied to
the existence of \emph{nontrivial} filtered subdeformation of a given
graded algebra.  We can fare better if we restrict the classification
problem to the Killing ideals.  In the next section, we consider the
odd-generated realizable case and we will show that one can fully specify
Killing ideals in terms of simpler data than $(\fh,S',\varphi, X)$
and, at the same time, modify step (1) by the addition of further
constraints.

\subsection{Killing ideals and Lie pairs}
\label{sec:LiePairs}

To state the main result of this section, we first need to introduce
some preliminary notions.  Let $S$ be the spinor representation of
$\fso(V)$. It is well-known that
\begin{equation}
  \label{eq:decompsymm}
  \odot^2 S\cong\Lambda^1 V\oplus \Lambda^2 V\oplus\Lambda^5 V,
\end{equation}
as $\fso(V)$-modules. This decomposition is unique, since all the
three summands are $\fso(V)$-irreducible and inequivalent, and we may
(and in this section will) consider $\Lambda^q V$ directly as a
subspace of $\odot^2 S$, for $q=1,2,5$. We decompose any element
$\omega\in\odot^2 S$ into
$ \omega=-\tfrac{1}{32}\big(\omega^{(1)}+\omega^{(2)}+\omega^{(5)}\big)$ according to
\eqref{eq:decompsymm}, where $\omega^{(q)}\in\Lambda^q V$ for
$q=1,2,5$. The overall factor of $-\tfrac{1}{32}$ is introduced so that $\omega^{(1)}$ coincides exactly with the Dirac current of $\omega$.

If $S'$ is a given linear subspace of $S$ with $\dim S'>16$, then
$\odot^2 S'\subset \odot^2 S$ projects surjectively on $\Lambda^1 V$,
through the Dirac current operation.  The embedding
\begin{equation*}
  \odot^2 S'\subset \odot^2 S=\Lambda^1V\oplus\Lambda^2
  V\oplus\Lambda^5 V
\end{equation*}
is in general diagonal and one cannot expect  $\odot^2 S'$ to contain
$\Lambda^q V$, not even if $q=1$. This motivates the following.

Let $S'$ be a subspace of $S$, $\dim S'>16$. Then restricting the
Dirac current $\kappa : \odot^2 S \to V$ to $\odot^2 S'$ gives rise to
a short exact sequence:
\begin{equation*}
  \begin{CD} 0 @>>> \fD @>>> \odot^2 S' @>\kappa>> V @>>> 0,
  \end{CD}
\end{equation*}
where $\fD = \fD(S')$ is the \emph{Dirac kernel} of $S'$; that is, the
subspace of $\odot^2 S$ given by
\begin{equation*}
  \begin{split} \mathfrak{D} &= \odot^2 S'\cap (\Lambda^2
    V\oplus\Lambda^5 V)\\ &=\left\{\omega\in\odot^2
      S'~\middle|~\omega^{(1)}=0\right\}.
  \end{split}
\end{equation*}
A splitting of the above short exact sequence --- that is, a linear
map $\Sigma:V\to \odot^2 S'$ such that $\Sigma(v)^{(1)}=v$ for all
$v\in V$ --- is called a \emph{section} associated to $S'$ and we may
write it as $\Sigma(S')$ if we need to specify $S'$.  A section
$\Sigma$ associated to $S'$ always exists and it is unique up to
elements in the Dirac kernel.

Let $(S',\varphi)$ be a pair consisting of a subspace $S'$ of $S$ with
$\dim S'>16$ and $\varphi\in\Lambda^4 V$.

\begin{definition}
  \label{def:envelope}
  \label{def:Lie pair}
  The \emph{envelope} $\fh_{(S',\varphi)}$ of $(S',\varphi)$ is the
  subspace of $\fso(V)$ given by
  \begin{equation*}
    \begin{split}
      \fh_{(S',\varphi)} &= \left\{\gamma^{\varphi}(\omega)
        ~\middle|~\omega\in\mathfrak{D}\right\}\\
      &= \left\{\gamma^{\varphi}(\omega)~\middle|~\omega\in\odot^2
        S'~\text{with}~\omega^{(1)}=0\right\}.
    \end{split}
  \end{equation*}
  The pair $(S',\varphi)$ is called a \emph{Lie pair} if
  \begin{enumerate}[label=(\roman*)]
  \item $A\cdot\varphi=0$ for every $A\in\fh_{(S',\varphi)}$; and
  \item $\sigma(A)s\in S'$ for every $A\in\fh_{(S',\varphi)}$ and $s\in
    S'$.
  \end{enumerate}
\end{definition}

The name ``Lie pair'' is motivated by the following

\begin{lemma}
  \label{lem:liepair}
  The envelope $\fh_{(S',\varphi)}$ of a Lie pair $(S',\varphi)$ is a
  Lie subalgebra of $\fso(V)$.
\end{lemma}

\begin{proof}
  The map $\gamma^{\varphi}:\odot^2 S\rightarrow \fso(V)$ is
  equivariant under $\fso(V)\cap\stab(\varphi)$, hence the restriction
  $\gamma^{\varphi}|_{\odot^2 S'}:\odot^2 S'\rightarrow \fso(V)$ to
  $S'$ is equivariant under
  $\fso(V)\cap\stab(\varphi)\cap\stab(S')$. Now
  $\fh_{(S',\varphi)}\subset\fso(V)\cap\stab(\varphi)\cap\stab(S')$ by
  properties (i) and (ii) of Definition \ref{def:Lie pair}. In
  particular, for any $A\in\fh_{(S',\varphi)}$ and
  $\omega\in\mathfrak{D}$, we have
  $ [A,\gamma^{\varphi}(\omega)]=\gamma^{\varphi}(A\cdot\omega) $,
  with $A\cdot\omega\in\mathfrak{D}$. In other words
  $[\fh_{(S',\varphi)},\fh_{(S',\varphi)}]\subset\fh_{(S',\varphi)}$,
  proving the lemma.
\end{proof}

The following result gives necessary conditions that are satisfied by
any odd-generated realizable filtered subdeformation. We recall that a
realizable $\fg=\fg_{\bar 0}\oplus\fg_{\bar 1}$ is called odd-generated
realizable if in addition $\fg_{\bar 0}=[\fg_{\bar 1},\fg_{\bar 1}]$.

\begin{proposition}
  \label{prop:LiePair}
  Let $\fg$ be a odd-generated realizable filtered subdeformation of $\fp$,
  with associated graded algebra $\fa=\fh\oplus S'\oplus V$. Then the
  associated pair $(S',\varphi)$ is a Lie pair and
  \begin{enumerate}[label=(\arabic*)]
  \item the isotropy $\fh$ equals the envelope of $(S',\varphi)$;
    i.e., $\fh=\fh_{(S',\varphi)}$;
  \item the map $X:V\to\fso(V)$ is determined, up to elements in
    $\fh$, by the identity
    \begin{equation*}
      X=\gamma^\varphi\circ\Sigma,
    \end{equation*}
    where $\Sigma$ is any section associated to $S'$.
  \end{enumerate}  
  In particular $\fg$ is completely determined, up to isomorphisms of
  filtered subdeformations, by the associated Lie pair $(S',\varphi)$.
\end{proposition}

\begin{proof}
  Let $T:=[-,-]|_{\odot^2 S'}:\odot^2 S'\to V\oplus\fh$ be the tensor
  given by the Lie bracket between odd elements. By equation
  \eqref{eq:preKSA}, it has the following explicit expression for all
  $\omega=-\tfrac{1}{32}\big(\omega^{(1)}+\omega^{(2)}+\omega^{(5)}\big)\in\odot^2 S'$:
  \begin{equation}
    \label{eq:formoftheoddbracket}
    \begin{split}
      T(\omega) &= \kappa(\omega) + \gamma^{\varphi}(\omega) - X_{\kappa(\omega)}\\
      &= \omega^{(1)} + \gamma^{\varphi}(\omega) - X_{\omega^{(1)}},
    \end{split}
  \end{equation}
  with $\omega^{(1)}\in V$ and
  $\gamma^{\varphi}(\omega) - X_{\omega^{(1)}}\in\fh$. The last
  identity in \eqref{eq:formoftheoddbracket} follows from the fact that
  the kernel of the Dirac current $\kappa:\odot^2 S\to V$ is
  $\Lambda^2 V\oplus\Lambda^5 V$.
  
  The tensor $T$ is surjective since $\fg$ is odd-generated.  In
  particular any $A\in\fh$ is of the form $A=T(\omega)$, for some
  $\omega\in\odot^2 S'$. By equation \eqref{eq:formoftheoddbracket},
  the condition $T(\omega)\in\fh$ is equivalent to $\omega^{(1)}=0$
  and hence $A=\gamma^{\varphi}(\omega)$ for some
  $\omega\in\mathfrak D$. In other words, $\fh=\fh_{(S',\varphi)}$,
  which proves~(1).
  
  Surjectivity of $T$ also allows one to choose (in a linear
  fashion) for every $v\in V$ an element $\Sigma(v)\in\odot^2 S'$
  with $T(\Sigma(v))=v$. Note that $\Sigma(v)^{(1)}=v$ by equation
  \eqref{eq:formoftheoddbracket}, i.e., $\Sigma:V\to\odot^2 S'$ is a
  section associated to $S'$. On the other hand
  $\gamma^{\varphi}(\Sigma(v))-X_v=0$
  for all $v\in V$, i.e., $X=\gamma^\varphi\circ\Sigma$.  Since
  sections associated to $S'$ differ by elements in $\mathfrak{D}$, a
  different choice of $\Sigma$ determines $X$ up to elements in
  $\fh=\fh_{(S',\varphi)}$. This proves (2).

  The fact that $(S',\varphi)$ is a Lie pair is a direct consequence of
  $\fh\subset\stab(S')\cap\stab(\varphi)$; the last claim of the
  proposition follows from (1), (2) and
  Section~\ref{sec:class-probl-highly}.
\end{proof} 

Proposition \ref{prop:LiePair} improves Proposition \ref{prop:fil1} in
the case of \emph{odd-generated realizable} filtered subdeformations. It
also allows to modify step (1) of the classification problem in
Section \ref{sec:class-probl-highly} with the following step:

\begin{enumerate}
\item[(1')] classify Lie pairs $(S',\varphi)$ (and therefore the
  corresponding graded algebras $\fa=\fh_{(S',\varphi)}\oplus S'\oplus
  V$), up to isomorphism.
\end{enumerate}

Here we say that two pairs
$(S',\varphi) \cong (g\cdot S', g\cdot\varphi)$ are isomorphic, where
$g\in\Spin(V)$. In this case
\begin{equation*}
  \mathfrak{D}(g\cdot S')=g\cdot \mathfrak{D}(S'),\qquad\Sigma(g\cdot
  S')=g\cdot\Sigma(S'),\qquad
  \fh_{(g\cdot S',g\cdot \varphi)}=g\cdot\fh_{(S',\varphi)}
\end{equation*}
and it is immediate that $(g\cdot S',g\cdot \varphi)$ is a Lie pair if
and only if $(S',\varphi)$ is a Lie pair.

\section{Towards the field equations}
\label{sec:field-eqns}

In this section we explore the possibility of deriving the field
equations from the Jacobi identity \eqref{eq:ssvJac}. The main result
is Theorem \ref{thm:mainII} in Section~\ref{sec:field-eqnsII}, which
states that if the Killing superalgebra is highly supersymmetric, then
the bosonic field equations are satisfied. We begin with some
preliminary results.  We shall only need some of the formulae in the
propositions below, but we record them all for completeness and
because one of the identities corrects a small error which has propagated in
the literature.

\subsection{The algebraic and differential relations}
\label{sec:field-eqnsI}

Let $(M,g,F)$ be an $11$-dimensional lorentzian spin manifold endowed
with a closed $4$-form $F\in\Omega^4(M)$. Associated to any spinor
field $\varepsilon\in\Gamma(\$)$, there are differential forms on $M$,
defined as follows:
\begin{enumerate}[label=(\roman*)]
\item $\omega^{(1)}\in\Omega^1(M)$, where
  $\omega^{(1)}(Z)=\left<\varepsilon,Z\cdot\varepsilon\right>$;
\item $\omega^{(2)}\in\Omega^2(M)$, where
  $\omega^{(2)}(Z_1,Z_2)=\left<\varepsilon,(Z_1\wedge
    Z_2)\cdot\varepsilon\right>$;
\item $\omega^{(5)}\in\Omega^5(M)$, where
  $\omega^{(5)}(Z_1,\ldots,Z_5)=\left<\varepsilon,(Z_1\wedge\ldots\wedge
    Z_5)\cdot\varepsilon\right>$.
\end{enumerate}
The $1$-form $\omega^{(1)}$ is the metric dual of the Dirac current
$K=\kappa(\varepsilon,\varepsilon)\in\mathfrak{X}(M)$ of $\varepsilon$.
Forms $\omega^{(q)}\in\Omega^q(M)$ can also be defined in a similar
way for $q=6,9,10$ and it is straightforward to check that they are
the Hodge duals $\omega^{(q)}=\star\omega^{(11-q)}$ of (i)-(iii).

The differential forms defined above satisfy certain algebraic
relations which are a consequence of the underlying Clifford
algebra. They are usually proved by repeated applications of Fierz
rearrangements.
\begin{proposition}(\cite[p. 5]{GauPak}, \cite[p. 21]{GauGutPak})
  \label{prop:algrel}
  Let $\varepsilon\in\Gamma(\$)$ be a spinor field, with associated
  Dirac current $K=\kappa(\varepsilon,\varepsilon)\in\mathfrak{X}(M)$. Then:
  \begin{align}
    \|\omega^{(2)}\|^2 &= 5\|\omega^{(1)}\|^2\\ 
    \|\omega^{(5)}\|^2 &= -6\|\omega^{(1)}\|^2\\
g(\imath_{Z}\omega^{(2)},\imath_{W}\omega^{(2)}) &= -\omega^{(1)}(Z)
                                                   \omega^{(1)}(W) +
                                                   g(Z,W) \|\omega^{(1)}\|^2\\ 
g(\imath_{Z}\omega^{(5)},\imath_{W}\omega^{(5)})&=14\omega^{(1)}(Z)\omega^{(1)}(W)-4g(Z,W)\|\omega^{(1)}\|^2\\
\imath_{K}\omega^{(1)}&=\|\omega^{(1)}\|^2\\
\imath_K\omega^{(2)}&=0\label{eq:alg6}\\
\imath_K\omega^{(5)}&=-\tfrac12 \omega^{(2)}\wedge\omega^{(2)}\label{eq:alg7}\\
\label{eq:wrong}
\imath_K\star\omega^{(5)}(Z_1,\ldots,Z_5)&=\sk_{Z_1,\ldots, Z_5} g(\imath_{Z_1}\omega^{(2)},\imath_{Z_2}\cdots\imath_{Z_5}\omega^{(5)})\\
\|\omega^{(1)}\|^2\omega^{(2)}\wedge\omega^{(5)}&=-\tfrac12\omega^{(1)}\wedge\omega^{(2)}\wedge\omega^{(2)}\wedge\omega^{(2)}\\
\notag
g(\imath_{Z_1}\omega^{(2)}, \imath_{Z_2}\cdots\imath_{Z_6}\star\omega^{(5)})&=5\sk_{Z_2,\ldots,Z_6}g(Z_1,Z_2)\imath_K\omega^{(5)}(Z_3,\ldots,Z_6)\\
&\quad {} - 5\sk_{Z_2,\ldots,Z_6}\omega^{(5)}(Z_1,Z_2,\ldots,Z_5)\omega^{(1)}(Z_6)\\
\omega^{(2)}\wedge\omega^{(2)}(Z_1,\ldots,Z_4) &= - \tfrac65 \sk_{Z_1,\ldots,Z_4}g(\imath_{Z_1}\imath_{Z_2}\omega^{(5)}, \imath_{Z_3}\imath_{Z_4}\omega^{(5)})
\end{align}
for all vector fields $Z,W,Z_i\in\mathfrak{X}(M)$, $i=1,\ldots,6$,
where $\sk$ is skew-symmetrisation with weight one.
\end{proposition}

Formulae in Proposition \ref{prop:algrel} are by no means
exhaustive. We note that some of our identities differ in sign from
those in \cite{GauPak} and \cite{GauGutPak}; this is due to our
conventions on the metric, which is ``mostly minus'', and Clifford
algebras.  Equation~\eqref{eq:wrong} corrects equation~(2.14) of
\cite{GauPak} and equation~(B.6) of \cite{GauGutPak}.

The covariant derivative of the differential forms were also
calculated in \cite{GauPak} and \cite{GauGutPak}. They are
summarised in the following.

\begin{proposition}(\cite[p. 6]{GauPak}, \cite[p. 5]{GauGutPak})
  \label{prop:difrel}
  Let $\varepsilon\in\Gamma(\$)$ be a Killing spinor on $(M,g,F)$,
  with associated Dirac current
  $K=\kappa(\varepsilon,\varepsilon)\in\mathfrak{X}(M)$. Then:
  \begin{align}
    \label{eq:cov1}
\nabla_{W}\omega^{(1)}(Z)&=\tfrac13\omega^{(2)}(\imath_{Z}\imath_{W}F)-\tfrac16\star\omega^{(5)}(Z\wedge W\wedge F)\\
\notag
\nabla_{W}\omega^{(2)}(Z_1,Z_2)&=-\tfrac13\omega^{(1)}(\imath_{W}\imath_{Z_1}\imath_{Z_2}F)-\tfrac13\omega^{(5)}(Z_1\wedge Z_2\wedge\imath_{W}F)\\
\notag
&\quad {} +\tfrac16\omega^{(5)}(W\wedge Z_1\wedge\imath_{Z_2}F)-\tfrac16\omega^{(5)}(W\wedge Z_2\wedge\imath_{Z_1}F)\\
\label{eq:cov2}
&\quad {} -\tfrac16g(W,Z_1)\omega^{(5)}(Z_2\wedge F)+\tfrac16g(W,Z_2)\omega^{(5)}(Z_1\wedge F)\\
\notag
\nabla_{W}\omega^{(5)}(Z_1,\ldots,Z_5)&=\tfrac53\sk_{Z_1,\ldots, Z_5}\star\omega^{(5)}(Z_1\wedge\ldots\wedge Z_4\wedge\imath_{Z_5}\imath_{W}F)\\
\notag
&\quad {} -\tfrac13\omega^{(2)}\wedge\imath_{W}F(Z_1,\ldots,Z_5)-\tfrac16\star \omega^{(1)}(Z_1\wedge\ldots\wedge Z_5\wedge W\wedge F)\\
\notag
&\quad {} -\tfrac{10}6\sk_{Z_1,\ldots, Z_5}\star\omega^{(5)}(Z_1\wedge Z_2\wedge Z_3\wedge\imath_{Z_4}\imath_{Z_5}(W\wedge F))\\
\label{eq:cov3}
&\quad {} -\tfrac56\sk_{Z_1,\ldots, Z_5}\omega^{(2)}(Z_1\wedge\imath_{Z_2}\ldots\imath_{Z_5}(W\wedge F)).
\end{align}
In particular the exterior differentials of the forms are given by:
\begin{align}
  \label{eq:diff1}
  d\omega^{(1)}&=\tfrac13\star(F\wedge\omega^{(5)}) + \tfrac23 \star(\star F\wedge\omega^{(2)})\\
  \label{eq:diff2}
  d\omega^{(2)} &=-\imath_K F\\
  \label{eq:diff3}
  d\omega^{(5)}&= \imath_K\star F - \omega^{(2)}\wedge F.
\end{align}
\end{proposition}

From Propositions \ref{prop:algrel} and \ref{prop:difrel} we can
immediately deduce the following result; the important identity
\eqref{eq:important} already appeared in \cite[p. 7]{GauPak}.

\begin{corollary}
  \label{cor:hsMax}
  Let $(M,g,F)$ be an $11$-dimensional lorentzian spin manifold
  endowed with a closed $4$-form $F\in\Omega^4(M)$. If
  $\varepsilon\in\Gamma(\$)$ is a Killing spinor, then the associated
  Dirac current $K\in\mathfrak{X}(M)$ is an $F$-preserving Killing
  vector which satisfies
  \begin{equation}
    \eL_K\omega^{(1)}=\eL_K\omega^{(2)}=\eL_K\omega^{(5)}=0
  \end{equation} 
  and
  \begin{equation}
    \label{eq:important}
    \imath_K(d\star F-\tfrac12F\wedge F)=0.
  \end{equation}
  In particular if the Killing superalgebra $\fk=\fk_{\bar
    0}\oplus\fk_{\bar 1}$ is highly supersymmetric, then $(M,g,F)$
  satisfies the Maxwell equation of $11$-dimensional supergravity.
\end{corollary} 

\begin{proof}
  The Dirac current $K$ is a Killing vector, since \eqref{eq:cov1} is
  evidently skew-symmetric in $W$ and $Z$. Moreover
  $\eL_K F=d\imath_K F+\imath _K dF=d\imath_ KF=0$, by $dF=0$ and
  \eqref{eq:diff2}. Equation $\eL_K\omega^{(1)}=0$ is immediate,
  whereas
  \begin{equation*}
    \begin{split}
      \eL_K\omega^{(2)}&=\imath_K d\omega^{(2)}=-\imath_{K}\imath_{K}F=0\\
      \eL_K\omega^{(5)}&=d\imath_K\omega^{(5)}+\imath_{K}d\omega^{(5)}\\
      &=-d\omega^{(2)}\wedge\omega^{(2)}-\imath_{K}(\omega^{(2)}\wedge F)\\
      &=\imath_{K}F\wedge\omega^{(2)}-\omega^{(2)}\wedge\imath_{K}
      F=0,
    \end{split}
  \end{equation*}
  using equations~\eqref{eq:alg6}, \eqref{eq:alg7}, \eqref{eq:diff2}
  and \eqref{eq:diff3}. Finally
  \begin{equation*}
    \begin{split}
      0=\star \eL_K F=\eL_K \star F&=d\imath_{K}\star F+\imath_{K}d\star F\\
      &=d(\omega^{(2)}\wedge F)+\imath_{K}d\star F=d\omega^{(2)}\wedge F+\imath_{K}d\star F\\
      &=-\tfrac12\imath_{K}(F\wedge F)+\imath_{K}d\star
      F=\imath_K(d\star F-\tfrac12F\wedge F),
    \end{split}
  \end{equation*}
  using \eqref{eq:diff2} and \eqref{eq:diff3}. The last claim is a
  direct consequence of \eqref{eq:important} and the surjectivity of
  the Dirac current.
\end{proof}

\subsection{High supersymmetry implies the field equations}
\label{sec:field-eqnsII}

The main result of this section is the following

\begin{theorem}
  \label{thm:mainII}
  Let $(M,g,F)$ be an $11$-dimensional lorentzian spin manifold
  endowed with a closed $4$-form $F\in\Omega^4(M)$.  If the associated
  Killing superalgebra $\fk=\fk_{\bar 0}\oplus\fk_{\bar 1}$ is highly
  supersymmetric (i.e., if $\dim\fk_{\bar 1}>16$) then $(M,g,F)$ is a
  supergravity background; i.e., it satisfies the bosonic field
  equations of $11$-dimensional supergravity:
  \begin{equation}
    \label{eq:bosfieldeqs}
    \begin{split}
      d\star F &= \tfrac12 F\wedge F,\\
      \Ric(Z,W) &=\tfrac12 g(\imath_ZF,\imath_W F) - \tfrac16
      \|F\|^2 g(Z,W),
    \end{split}
  \end{equation}
  for all $Z,W\in\mathfrak{X}(M)$.
\end{theorem}

The proof of Theorem \ref{thm:mainII} will occupy the remainder of
this section, but before we start let us remark that the theorem is
sharp.  Indeed, there exist lorentzian $11$-dimensional manifolds
$(M,g)$ with $F=0$, which admit a 16-dimensional space of parallel
spinors and which are not Ricci-flat \cite{JMWaves,Bryant-ricciflat}.

Let us now turn to the proof of the theorem.  From now on, we will use
the Einstein summation convention and consider the canonical
isomorphism $\Lambda^\bullet V\cong\Cl(V)$ of vector spaces. It sends
a $p$-polyvector
$\Theta=\tfrac1{p!}\Theta^{a_1\cdots a_p}e_{a_1}\wedge\cdots\wedge
e_{a_p}$ into
$\tfrac1{p!}\Theta^{a_1\cdots a_p}\Gamma_{a_1\ldots a_p}$, where
$(e_a)$ is any $\eta$-orthonormal basis of $V$ and
$\Gamma_{a_1\ldots a_p}$ the totally antisymmetric product (with
weight one) of the corresponding operators $\Gamma_{a_i}\in\Cl(V)$ of
Clifford multiplication by $e_{a_i}\in V$. Finally, we denote by
$[\Xi]_p$ the $p$-form component of $\Xi \in \Cl(V)$.

We begin with two useful lemmas.

\begin{lemma}
  \label{lem:use2}
  Let $\Theta\in\Lambda^p V$ be a $p$-polyvector. Then
  \begin{equation}
    \label{eq:use2I}
    u \cdot \Theta = u \wedge \Theta - \iota_u \Theta
    \qquad\text{and}\qquad
    \Theta \cdot u = (-1)^p \left(u \wedge \Theta + \iota_u \Theta \right),
  \end{equation} 
  for all $u \in V$. In particular 
  \begin{equation}
    \label{eq:use2II}
    \tr_{v,w} v\cdot \Theta\cdot w = (-1)^{p+1} (11-2p) \Theta,
  \end{equation}
  where $\tr_{v,w}$ is the tracing operation over $v,w\in V$. 
\end{lemma}

\begin{lemma}
  \label{lem:use1}
  Let $\varphi\in\Lambda^4 V$ be a $4$-polyvector. Then 
  \begin{equation}
    \label{eq:use1I}
    u \wedge \varphi = \tfrac12 (u\cdot \varphi + \varphi\cdot u ) \qquad\text{and}\qquad
    \iota_u \varphi = \tfrac12 (\varphi\cdot u - u\cdot \varphi),
  \end{equation}
  for all $u\in V$. Moreover
  \begin{equation}
    \label{eq:use1II}
    \varphi^2 = \varphi\cdot\varphi=\varphi \wedge \varphi + [\varphi^2]_4 + \|\varphi\|^2 \1,
  \end{equation}
  where $[\varphi^2]_4 = - \tfrac18 \varphi^{abmn}\varphi_{mn}{}^{cd} \Gamma_{abcd}$.
\end{lemma}

The identities in Lemmas~\ref{lem:use2}~and~\ref{lem:use1} are
obtained by routine calculations in $\Cl(V)$. We omit the proof for
the sake of brevity.

So let us now fix a point $o\in M$ and assume $\dim S'> 16$, so that
$\kappa: \odot^2 S' \to V$ is surjective and $V'=V$.  We will abuse
notation slightly by using $F$ both for the $4$-form as for the value
at $o$, which is an element of $\Lambda^4 V\cong\Lambda^4 V^*$.

We consider the $[S'S'V]$ Jacobi identity \eqref{eq:ssvJac} and take
the inner product with a vector $u \in V$ to arrive at
\begin{equation*}
 \eta(u,R(v,\kappa(s,s))w) = 2 \left<u \cdot \beta^F_w \cdot \beta^F_v
   \cdot s, s\right> + 2 \left<u \cdot \beta^F_w \cdot s, \beta^F_v \cdot
   s \right> - 2 \left<u \cdot (X_v\beta^F)(w,s),s\right>\,,
\end{equation*}
for all $u, v,w\in V$ and $s\in S$. Now the symplectic transpose of 
$\beta^F_v = \tfrac1{24} (v \cdot F - 3 F \cdot v)$ is 
$\widetilde{\beta^F_v} = \tfrac1{24} (3 v \cdot
F - F \cdot v)$, so that
\begin{equation*}
 \eta(u,R(v,\kappa(s,s))w) = 2 \left<\left(u \cdot \beta^F_w \cdot \beta^F_v
   + \widetilde{\beta^F_v} \cdot u \cdot \beta^F_w \right) \cdot s,
 s\right> - 2 \left<u \cdot (X_v\beta^F)(w,s),s\right>.
\end{equation*}
We now expand by using the definition of $\beta^F$ and the fact that
\begin{equation*}
  (X_v\beta^F)(w,s) = \beta^{X_vF}(w,s)=\tfrac1{24} (w \cdot
(X_v F) - 3 (X_v F) \cdot w)\cdot s, 
\end{equation*}
for all $v,w\in V$ and $s\in S$. Dropping the Clifford multiplication
$\cdot$ from the notation, we arrive at
\begin{equation*}
  \begin{split}
    \eta(u,R(v,\kappa(s,s))w) & = \tfrac2{(24)^2}
    \bigl<\bigl(uwF vF -
        3uF wvF - 3uwF^2v + 9 u F wF v  + 3 vF uwF\\
    & \quad {} - 9 vF uF w - F vuwF +
        3F vuF w \bigr) \cdot s, s\bigr>\\ & \quad {} -
    \tfrac1{12}  \left<(u w (X_v F) - 3 u (X_v F) w)\cdot s, s\right>.
  \end{split}
\end{equation*}
The Ricci tensor is obtained by ``tracing'' over $v,w$ and taking the opposite:
\begin{equation}
\label{eq:RicciI}
  \begin{split}
    \Ric(u,\kappa(s,s)) &= -\tr_{v,w} \eta(u,R(v,\kappa(s,s))w)\\
    &= -\tfrac2{(24)^2} \left< \Upsilon_u \cdot s, s\right> +
    \tfrac1{12} \tr_{v,w} \left<(u w (X_v F) - 3 u (X_v F) w)\cdot s,
      s\right> ,
  \end{split}
\end{equation}
where
\begin{multline*}
  \Upsilon_u = \tr_{v,w}\left(uwF vF -
    3uF wvF - 3uwF^2v + 9 u F wF v \right. \\
    \left. {} + 3 vF uwF - 9 vF uF w - F
      vuwF  + 3F vuF w  \right).
\end{multline*}
We treat the two terms in the RHS of \eqref{eq:RicciI} separately and
in turn. First we expand the $\Upsilon$ term by making use of
\eqref{eq:use1II} in Lemma~\ref{lem:use1} and the following traces,
which are a direct consequence of \eqref{eq:use2II} in
Lemma~\ref{lem:use2}:
\begin{equation*}
  \begin{split}
    \tr_{v,w} v F w &= -3 F\\
    \tr_{v,w} v w &= -11 \1\\
    \tr_{v,w} v u w &= 9 u \\
    \tr_{v,w} v F^2 w &= 5 F \wedge F - 3
    [F^2]_4 - 11 \|F\|^2 \1.
  \end{split}
\end{equation*}
Therefore substituting this into $\Upsilon_u$ we find
\begin{multline*}
  \Upsilon_u = -12 u (F\wedge F) + 12 u [F^2]_4 + 36 \|F\|^2 u +
  3 (u\wedge F) F + 3 F (u \wedge F)\\
  {} + 15 (\iota_u F) F
  - 15 F (\iota_u F) - 9 F u F - 9 \tr_{v,w} v(F u F) w.
\end{multline*}
Remember, though, that this expression appears in
\begin{equation*}
\left<\Upsilon_u \cdot s, s\right> = - \left<s, \Upsilon_u \cdot s\right> = -
\left<\widetilde{\Upsilon_u} \cdot s, s\right>,
\end{equation*}
where $\widetilde{\Upsilon_u}$ is the symplectic transpose of $\Upsilon_u$, so that
\begin{equation*}
  \left<\Upsilon_u \cdot s, s\right> = \tfrac12
  \left<(\Upsilon_u-\widetilde{\Upsilon_u})\cdot s, s\right>.
\end{equation*}
Using that for $\Theta$ a $p$-polyvector, $\widetilde\Theta =
(-1)^{p(p+1)/2} \Theta$, we may thus replace $\Upsilon_u$ by the following term
\begin{multline*}
 -12 u (F\wedge F) + 12 u [F^2]_4 + 36 \|F\|^2 u +
  6 (u\wedge F) F \\
  {} + 30 (\iota_u F) F - 9 F u F - 9
  \tr_{v,w} v(F u F) w.
\end{multline*}
Identities \eqref{eq:use2I} in Lemma \ref{lem:use2} allows to further
expand this term, and keeping in mind that only the skewsymmetric
endomorphisms survive, we arrive at
\begin{multline*}
  \tfrac1{24} (\Upsilon_u- \widetilde{\Upsilon_u}) = 4 u \wedge
  F \wedge F + u \wedge [F^2]_4 +
  3 \|F\|^2 u - [(u\wedge F)F]_5\\
	- 7 [(u\wedge F) F]_1 + [(\iota_u F)
        F]_5 - 5 [(\iota_u F) F]_1\,.
\end{multline*}
Now observe that the 2nd, 4th and 6th terms add to zero, so that
\begin{equation*}
  \tfrac1{12} \left< \Upsilon_u \cdot s, s\right> =
  \left< (4 u \wedge F \wedge F + 3 \|F\|^2 u - 7
    [(u\wedge F) F]_1 -5 [(\iota_u F)F]_1 ) \cdot s,
    s\right>
\end{equation*}
and, from \eqref{eq:use1I} in Lemma \ref{lem:use1}, we arrive at
\begin{equation*}
  \tfrac1{12}\left< \Upsilon_u \cdot s, s\right>   =
  \left< (4 u \wedge F \wedge F + 2 \|F\|^2 u - 6 [F uF]_1)
    \cdot s, s\right>.
\end{equation*}
It is clear after a moment's thought that
\begin{equation*}
  [F u F]_1 = \left(\alpha \|F\|^2 \eta_{ab} + \beta
    F^2_{ab}\right) u^a \Gamma^b,
\end{equation*}
for some $\alpha,\beta \in \RR$, where
\begin{equation*}
  F^2_{ab}
  =\eta(\imath_{e_a}F,\imath_{e_b}F)=\tfrac16
  F_{amnp} F_b{}^{mnp}.
\end{equation*}
By taking $F = \Gamma_{0123}$ and taking $u= \Gamma_0$ and
$u=\Gamma_5$ in turn, say, we find that $\alpha = 1$ and $\beta = -2$,
so that in the end
\begin{equation}
\label{eq:riccipart1}
    \tfrac2{(24)^2} \left< \Upsilon_u \cdot s, s\right> = \tfrac16
    \left< u \wedge F \wedge F \cdot s, s\right> + \tfrac12 F^2_{ab}
    u^a \left<\Gamma^b s,s\right> - \tfrac16 \|F\|^2 \left<u \cdot s,
      s\right>.
\end{equation}

Now we tackle the other terms in \eqref{eq:RicciI}. We first observe
that $(v,X_v)$ is a Killing vector field which preserves $F$, by the
geometric interpretation of the Killing superalgebra in
Sections~\ref{sec:kill-super-supersymm}~and~\ref{sec:admfilsubdef}.  In
particular the Lie derivative $\eL_{(v,X_v)} F= \nabla_v F + X_vF = 0$
and hence $X_v F = - \nabla_vF$. Therefore,
\begin{equation*}
  \tr_{v,w} w \cdot (X_v F) = - d F + \delta F \qquad\text{and}\qquad
  \tr_{v,w} (X_v F) \cdot w = - dF - \delta F,
\end{equation*}
where $dF$ is the exterior derivative and $\delta F =-\star d\star F$
the divergence.  It follows that
\begin{equation*}
  - \tfrac1{12} \tr_{v,w} \left<(u w (X_v F) - 3 u (X_v F)
    w)\cdot s, s\right>  = - \tfrac16 \left< u \cdot (dF + 2
    \delta F) \cdot s, s\right>
\end{equation*}
and remembering that only the $1$-, $2$- and $5$-form terms (and their
duals) survive, we finally arrive at
\begin{equation}
  \label{eq:riccipart2}
  - \tfrac1{12} \tr_{v,w} \left<(u w (X_v F) - 3 u (X_v F)
    w)\cdot s, s\right>  = - \tfrac16 \left< (u \wedge dF - 2
    \iota_u \delta F) \cdot s, s\right>.
\end{equation}

In summary, we add (the opposite of) equations~\eqref{eq:riccipart1} and
\eqref{eq:riccipart2} to arrive at
\begin{multline}
  \label{eq:ricci}
  \Ric(u,\kappa(s,s)) = -\tfrac12 F^2_{ab} u^a \left<\Gamma^b s,s\right> +
  \tfrac16 \|F\|^2 \left<u \cdot s, s\right>\\
  - \tfrac16 \left< (u \wedge F \wedge F - u \wedge dF + 2
    \iota_u \delta F) \cdot s, s\right>.
\end{multline}
There are three kinds of terms which depend on
$s$ in equation~\eqref{eq:ricci}: terms which depend via the Dirac
current, terms which depend via the 2-form bilinear
$\omega^{(2)}$ and terms which depend via the
$5$-form bilinear
$\omega^{(5)}$ (see Section~\ref{sec:field-eqnsI} for definitions). The
embedding $\odot^2 S'\subset \odot^2 S=\Lambda^1 V\oplus \Lambda^2
V\oplus\Lambda^5
V$ is in general diagonal, and the fact that \eqref{eq:ricci} has to
be true for all $s\in
S'$ does not guarantee a priori that each of these three terms
satisfies the equation separately; although they do in the maximally
supersymmetric case when $S' = S$.

Notice however that the equation for the terms depending on the
$5$-form bilinear is
\begin{equation}
\label{eq:preBianchi}
  \left< (u \wedge dF) \cdot s, s\right> = 0,
\end{equation}
for all $u \in V$ and $s \in S'$. Similarly the equation for the terms
depending on the $2$-form (or, dually, the $9$-form) bilinear is
\begin{equation}
  \label{eq:preMaxwell}
  \left< (u \wedge F \wedge F + 2 \iota_u \delta F) \cdot s, s\right>
  = 0,
\end{equation}
for all $u\in V$ and $s\in S'$. By hypothesis $dF=0$, so that
\eqref{eq:preBianchi} is automatically satisfied. By high
supersymmetry and Corollary \ref{cor:hsMax}, the Maxwell equation of
$11$-dimensional supergravity is also satisfied and this directly
implies equation \eqref{eq:preMaxwell}.  This then boils down
equation~\eqref{eq:ricci} to the vanishing of the terms depending
just on the Dirac current, namely:
\begin{equation}
\label{eq:einstein}
  \Ric(u,\kappa(s,s)) = \tfrac12 F^2_{ab} u^a \left< s, \Gamma^b s\right> -
  \tfrac16 \|F\|^2 \left< s,u \cdot s\right>,
\end{equation}
which, since $\kappa$ is surjective, is nothing but the expected
Einstein equation $\Ric_{ab} = \tfrac12 F^2_{ab} - \tfrac16 \|F\|^2
g_{ab}$. Theorem \ref{thm:mainII} is hence proved.
\medskip\par
As a corollary, we now show that the space $\ker \imath^*$ given in Lemma
\ref{lem:kernelpriori} vanishes if $\dim S'>16$.

\begin{corollary}
  \label{cor:final!}
  Let $\fa=\fh\oplus S'\oplus V$ be a highly supersymmetric graded
  subalgebra of $\fp$. Then $\ker \imath^*=0$. In
  particular a filtered deformation $\fg$ of $\fa$ has at most one
  admissible $\varphi\in\Lambda^4 V$.
\end{corollary}

\begin{proof}
  We first note that
  $\imath^*:H^{2,2}(\fp_-,\fp)\to H^{2,2}(\fa_-,\fp)$ depends only on
  the negatively graded part $\fa_-=S'\oplus V$ of $\fa$. We can
  therefore assume without any loss of generality that $\fa=\fa_-$
  from now on, so that $\fh=0$.

  Now let $\varphi\in\Lambda^4 V$ such that the corresponding class
  $[\beta^\varphi+\gamma^\varphi]\in H^{2,2}(\fp_-,\fp)$ satisfies
  $\imath^*[\beta^\varphi+\gamma^\varphi]=0$. In other words
  $\beta^\varphi|_{V\otimes S'}=\gamma^\varphi|_{S'\otimes
    S'}=0$. Also let $\fg$ be the filtered deformation of $\fa$
  determined by the brackets \eqref{eq:generalbrackets} and
  \eqref{eq:generalbracketsII} with $X=\rho=0$; by construction $\fg$
  is a trivial realizable filtered subdeformation of $\fp$, with
  associated admissible $4$-polyvector $\varphi$. Triviality here
  refers to the fact that $\fg\cong\fa$ is actually a graded Lie
  subalgebra of $\fp$.

  It follows from Theorem \ref{thm:firstmain} that the associated
  homogeneous lorentzian spin manifold $(M,g,Q,F)$, where
  $F_o=\varphi$, has vanishing Riemann curvature. It is also highly
  supersymmetric so that, by Theorem \ref{thm:mainII}, it solves the
  bosonic field equations. In particular, the Einstein equation says
  \begin{equation}
    \label{eq:lasteq!}
    0=\tfrac12g(\imath _ZF,\imath_W F)-\tfrac16\|F\|^2 g(Z,W),
  \end{equation}
  for all $Z,W\in\mathfrak{X}(M)$. Taking the trace over $Z,W$ yields
  $0=\tfrac16\|F\|^2$ so that both terms in \eqref{eq:lasteq!} have
  to vanish separately and $g(\imath _ZF,\imath_W F)=0$ for all
  $Z,W\in\mathfrak{X}(M)$. Using a Witt basis for $T_oM$ it is then
  straightforward to see that this can only happen when
  $\varphi=F_o=0$.
\end{proof}

As we have had ample opportunity to see, filtered deformations $\fg$
of graded subalgebras $\fa$ of $\fp$ are not, in general, graded Lie
subalgebras of $\fp$. By Corollary~\ref{cor:final!}, the unique
\emph{highly supersymmetric} background associated to graded
subalgebras of $\fp$ is actually the Minkowski vacuum. In particular,
the Minkowski vacuum is also the unique \emph{highly supersymmetric}
background with vanishing flux $F$.

Corollary~\ref{cor:final!} fails to hold in the general case.  There
are indeed other supergravity backgrounds whose associated Killing
superalgebras are graded subalgebras of $\fp$. This is the case for
some $\tfrac12$-BPS solutions such as M2 and M5 branes, see
e.g., \cite{JMFSuperDeform}, and it  also seems to be the case for
backgrounds asymptotic to the Minkowski vacuum. Finally, any
Ricci-flat $11$-dimensional lorentzian spin manifold endowed with a
parallel spinor provides a low supersymmetric background with
vanishing flux $F$, cf.~\cite{JMWaves,Bryant-ricciflat}.

\section{Summary and conclusions}
\label{sec:summary-conclusions}

In this paper we have elucidated the algebraic structure of the Lie
superalgebra generated by the Killing spinors of an $11$-dimensional
supergravity background.  We have shown that it is a filtered
deformation of a $\ZZ$-graded subalgebra of the Poincaré superalgebra.
(Parenthetically, this is not unique to $11$-dimensional supergravity,
but it is known to be the case for the Lie algebra of automorphisms of
riemannian and conformal manifolds, as well as other supergravity
theories.  Moreover it is also expected to be the case for conformal
supergravities.)  Together with the (local) homogeneity theorem, which
states that ``highly supersymmetric'' backgrounds (i.e, those
preserving more than half of the supersymmetry) are locally
homogeneous, this provides a new approach to the classification
problem based on the classification of the Killing superalgebras (or
the Killing ideals) of such backgrounds, which we have identified with
a class of (odd-generated) realizable filtered subdeformations of the
Poincaré superalgebra.  We have outlined in purely algebraic terms the
classification problem of Killing ideals of highly supersymmetric
supergravity backgrounds.  It consists of two steps
\begin{enumerate}
\item classify all the Lie pairs $(S',\varphi)$ up to isomorphism; and
\item for each such isomorphism class, consider all $(R,X)$, where $R$
  is an ($\fh = \fh_{(S',\varphi)}$)-invariant algebraic curvature
  tensor and $X: V \to \fso(V)/\fh$, such that
  \begin{enumerate}[label=(\roman*)]
  \item the $4$-form $F$ defined by $\varphi$ is closed;
  \item the right-hand sides of \eqref{eq:preKSA} take values in
    $\fh\oplus S'\oplus V$; and
  \item the three equations~\eqref{eq:ssvJac}, \eqref{eq:svvJac} and
    \eqref{eq:difBianchiVp=V} are satisfied.
  \end{enumerate}
\end{enumerate}

Among the corollaries derived from this approach is the statement that
high supersymmetry (and $dF=0$) imply the bosonic field equations.
Hence we can be sure that classifying (maximal, odd-generated) realizable
filtered subdeformations one classifies highly supersymmetric
backgrounds. 

\section*{Acknowledgments}

We are grateful to the anonymous referee for comments on a previous
version of the paper.  We believe that the paper has improved as a
result of the peer-review process.  The present version of the paper
was finished while we were participating in the workshop ``Geometry,
Gravity and Supersymmetry'' at the Mainz ITP.  It is our pleasure to
thank them for their hospitality and for providing a pleasant
collaborative environment.

The research of JMF is supported in part by the grant ST/L000458/1
``Particle Theory at the Higgs Centre'' from the UK Science and
Technology Facilities Council.  The research of AS is fully supported
by a Marie-Curie research fellowship of the ''Istituto Nazionale di
Alta Matematica'' (Italy).  We are grateful to these funding agencies
for their support.


\providecommand{\href}[2]{#2}\begingroup\raggedright\endgroup

\end{document}